\AtBeginDocument{%
	\paperwidth=\dimexpr
	1in + \oddsidemargin
	+ \textwidth
	+ 1in + \oddsidemargin
	\relax
	\paperheight=\dimexpr
	1in + \topmargin
	+ \headheight + \headsep
	+ \textheight
	+ 1in + \topmargin
	\relax
	\usepackage[pass]{geometry}\relax
} 

\RequirePackage{fix-cm}
\documentclass{svjour3}                     
\pdfoutput=1 
\usepackage{wrapfig,epsf}

\smartqed  
\usepackage{mathptmx} 
\usepackage{graphicx}
\usepackage[ruled,vlined]{algorithm2e}
\usepackage[caption=false]{subfig}
\usepackage{float}
\usepackage{xspace}
\usepackage{paralist}
\usepackage{amssymb}
\usepackage{amsmath}
\usepackage{amsfonts}
\usepackage{mathtools}
\usepackage{url}
\usepackage{booktabs} 
\usepackage[misc,geometry]{ifsym} 
\DeclarePairedDelimiter\ceil{\lceil}{\rceil}

\DeclarePairedDelimiter\abs{\lvert}{\rvert}
\DeclarePairedDelimiter\norm\lVert\rVert

\DeclareMathAlphabet{\mathcal}{OMS}{cmsy}{m}{n}

\newcommand{\eg}[0]{\emph{e.g.,}\xspace}
\newcommand{\ie}[0]{\emph{i.e.,}\xspace}
\newcommand{\cf}[0]{\emph{cf.}\xspace}
\newcommand{\etal}[0]{\emph{et~al.}\xspace}

\newcommand{\term}[1]{\emph{#1}}
\newcommand{\eat}[1]{}

\newcommand{\zero}[0]{\ensuremath{\mathbf{0}}\xspace}

\newcommand{\Hb}[0]{\ensuremath{\mathbf{H}}\xspace} 
\newcommand{\Hn}[0]{\ensuremath{\mathbf{H'}}\xspace} 

\newcommand{\Fset}[0]{\ensuremath{\mathcal{F}}\xspace} 
\newcommand{\Eset}[0]{\ensuremath{\mathcal{E}}\xspace} 
\newcommand{\Vset}[0]{\ensuremath{\mathcal{V}}\xspace} 
\newcommand{\Hset}[0]{\ensuremath{\mathcal{H}}\xspace} 

\newcommand{\Hc}[0]{\ensuremath{\mathbf{H''}}\xspace} 

\newcommand{\Hr}[0]{\ensuremath{\mathbf{H'''}}\xspace} 

\newcommand{\h}[1]{\ensuremath{H_{#1}}\xspace} 

\newcommand{\hn}[1]{\ensuremath{H'_{#1}}\xspace} 

\newcommand{\hc}[1]{\ensuremath{H''_{#1}}\xspace} 
\newcommand{\fc}[1]{\ensuremath{F''_{#1}}\xspace} 
\newcommand{\ec}[1]{\ensuremath{E''_{#1}}\xspace} 
\newcommand{\vc}[1]{\ensuremath{V''_{#1}}\xspace} 

\newcommand{\hr}[1]{\ensuremath{H'''_{#1}}\xspace} 

\newcommand{\M}[0]{\ensuremath{\mathcal{M}}\xspace}

\newcommand{\Prob}[1]{\ensuremath{\mathrm{Pr}(#1)}}

\newcommand{\DP}[0]{\emph{DiffPriv}\xspace}
\newcommand{\LP}[0]{\emph{LinProg}\xspace}
\newcommand{\Eu}[0]{\emph{Euler}\xspace}
\newcommand{\R}[0]{\emph{Round}\xspace}

\newcommand{\dep}[0]{distinct-entry counting problem\xspace}
\newcommand{\dop}[0]{distinct-object counting problem\xspace}

\newcommand{\alice}{\ensuremath{\mathcal{A}}\xspace}
\newcommand{\bob}{\ensuremath{\mathcal{B}}\xspace}

\newtheorem{constraint}{Constraint}

\journalname{Knowledge and Information Systems (KAIS)}

\begin{document}
	\sloppy

\title{Differentially-Private Counting of Users' Spatial Regions}


\author{Maryam Fanaeepour  \and Benjamin I. P. Rubinstein }

\institute{Maryam Fanaeepour (\Letter) \at
	School of Computing and Information Systems,\\
	University of Melbourne, Parkville, VIC 3052, Australia\\
	Data61, CSIRO, Australia\\
	\email{maryamf@student.unimelb.edu.au}           
	\and
	Benjamin I. P. Rubinstein \at
	School of Computing and Information Systems,\\
	University of Melbourne, Parkville, VIC 3052, Australia\\
	\email{brubinstein@unimelb.edu.au}
}
 
\date{Received: date / Accepted: date}
\maketitle

\begin{abstract}
Mining of spatial data is an enabling technology for mobile services, Internet-connected cars, and the Internet of Things. But the very distinctiveness of spatial data that drives utility, can cost user privacy. Past work has focused upon points and trajectories for differentially-private release. In this work, we continue the tradition of privacy-preserving spatial analytics, focusing not on point or path data, but on planar spatial regions. Such data represents the area of a user's most frequent visitation---such as ``around home and nearby shops''. Specifically we consider the differentially-private release of data structures that support range queries for counting users' spatial regions.  
Counting planar regions leads to unique challenges not faced in existing work.
A user's spatial region that straddles multiple data structure cells can lead to duplicate counting at query time. We provably avoid this pitfall by leveraging the Euler characteristic for the first time with differential privacy. To address the increased sensitivity of range queries to spatial region data, we calibrate privacy-preserving noise using bounded user region size and a constrained inference that uses robust least absolute deviations. Our novel constrained inference reduces noise and promotes covertness by (privately) imposing consistency. We provide a full end-to-end theoretical analysis of both differential privacy and high-probability utility for our approach using concentration bounds. A comprehensive experimental study on several real-world datasets establishes practical validity.
\end{abstract}

\keywords{Differential Privacy \and Euler Histograms \and Location Privacy \and Spatial~Regions}

\section{Introduction}\label{sec:intro}

The ubiquity, quality and usability of location-based services supports the ready availability of user tracking.
Location data sharing is used across a wide range of applications such as traffic monitoring, facility location planning, recommendation systems and contextual advertising. The distinctiveness of location data, however, has led to calls for location privacy~\cite{Beresford03,Gruteser2004}: the ability to track users in aggregate without breaching individual privacy. There exists a spectrum of approaches to address location privacy~\cite{Chow2011,Chow2011Trajectory,Ghinita2013,Krumm07} with significant attention having been paid to range queries on point location or trajectory data: for example, providing statistics of how many mobile users are presently on an arterial road. 

Typical private spatial analytics supports point locations or sequences of points (see Figure~\ref{fig:map_points}).
Points and trajectories, however, do not best-represent user location in all applications. 
In facility-services planning, a planner may wish to locate a new department store in a location that overlaps with users' regions of frequent visitation. While hotel-booking sites collect area-level information about customers' preferred destinations. 
Such problems motivate our focus on counting private planar bodies\footnote{We use \term{body} and \term{region} interchangeably to refer to a user's spatial area. We use the term \term{body} to distinguish query regions from users' regions.} (see Figure~\ref{fig:map_bodies}). 
Given a collection of privacy-sensitive planar bodies representing regions of frequent location, we wish to support counting range queries while preserving individual privacy.
Figure~\ref{fig:map_bodies} illustrates this task, on a map of metropolitan Melbourne with planar bodies representing regions of individual users' frequent visitation. 
Third parties may wish to submit any number of queries requesting the number of users' areas falling in a specified query region, \eg for urban transport planning or retail analytics. 

A leading approach for responding to range queries in spatial data analytics is aggregation~\cite{Papadias2001,Papadias2002,Tao2002,Tao2004,Lopez2005,Braz2007,Marketos2008,LeonardiORRSAA14,Timko2009}. Initial interest in aggregation was due to computational efficiency considerations and early data structures promote these properties. More recently aggregation has been used as a qualitative approach to privacy, as it is a natural choice for privacy-preserving data release~\cite{Chawla2005}.

In the setting of planar bodies, conventional grid-partitioned histograms 
cannot provide accurate results due to the \term{duplicate counting}\footnote{In the literature, the terms \term{multiple}, \term{double} or \term{distinct} counting are used interchangeably. We suggest the term ``duplicate" as it conveys that objects are over-counted.} problem as a planar body may span more than one histogram cell simultaneously. This is a problem unique to counting planar bodies. To address this challenge, we instead leverage the Euler characteristic~\cite{trudeau1993introduction} 
where face, edge and vertex counts are stored separately. Such Euler histograms~\cite{Beigel1998} permit exact counting of convex planar bodies~\cite{SunAA02,Sun2002,Sun2006} (\cf Section~\ref{sec:prelim-euler} and Figure~\ref{fig:euler_histograms}).

\begin{figure*}
	\centering
	\subfloat[ Point locations or sequence of points (trajectories) as typical spatial representatives of a user.\label{fig:map_points}]{
		\includegraphics[scale = .14]{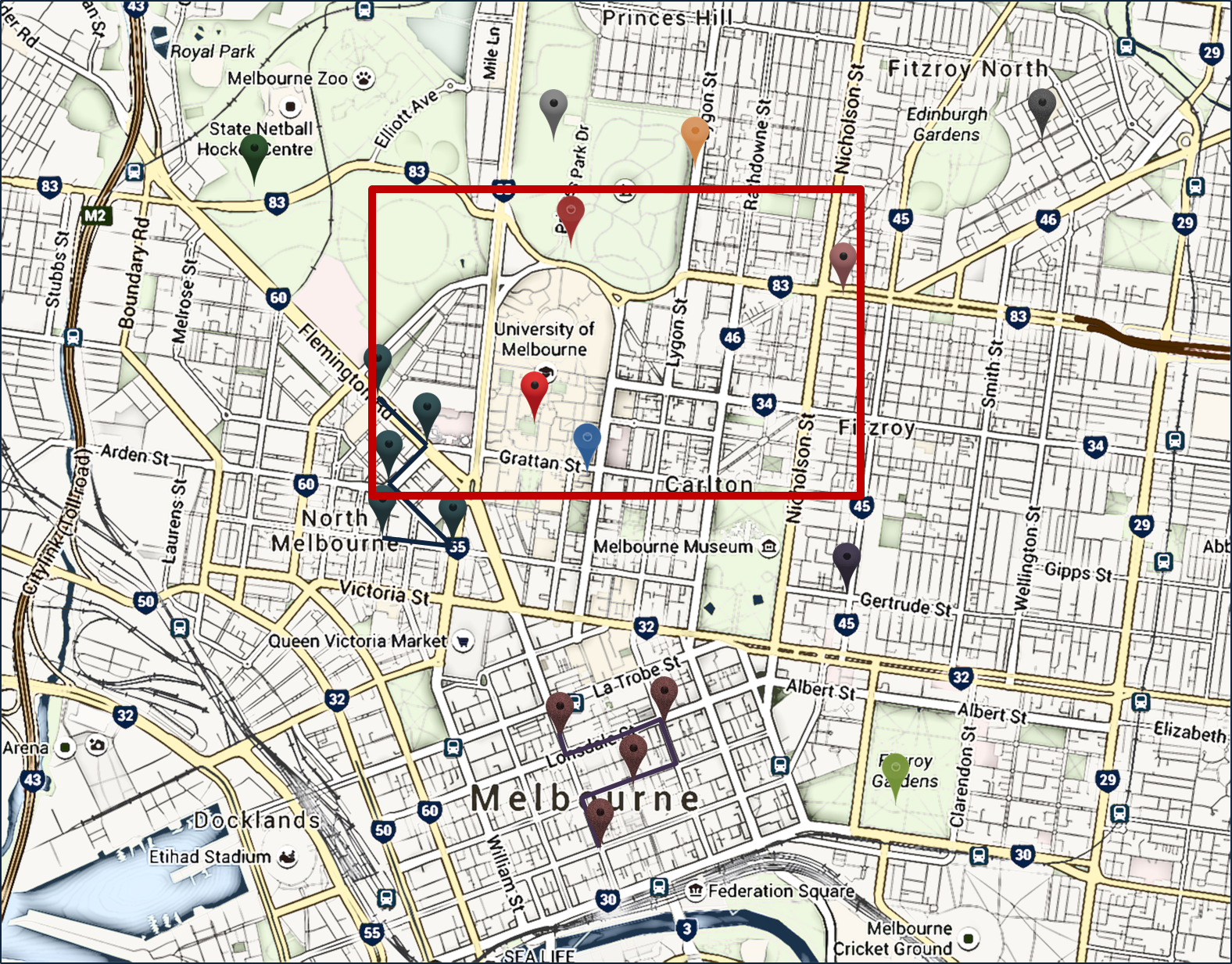}}
	\hfill
	\subfloat[ Users' regions of frequent visitation; each user's spatial data is represented by a single planar body.\label{fig:map_bodies}]{
		\includegraphics[scale = .14]{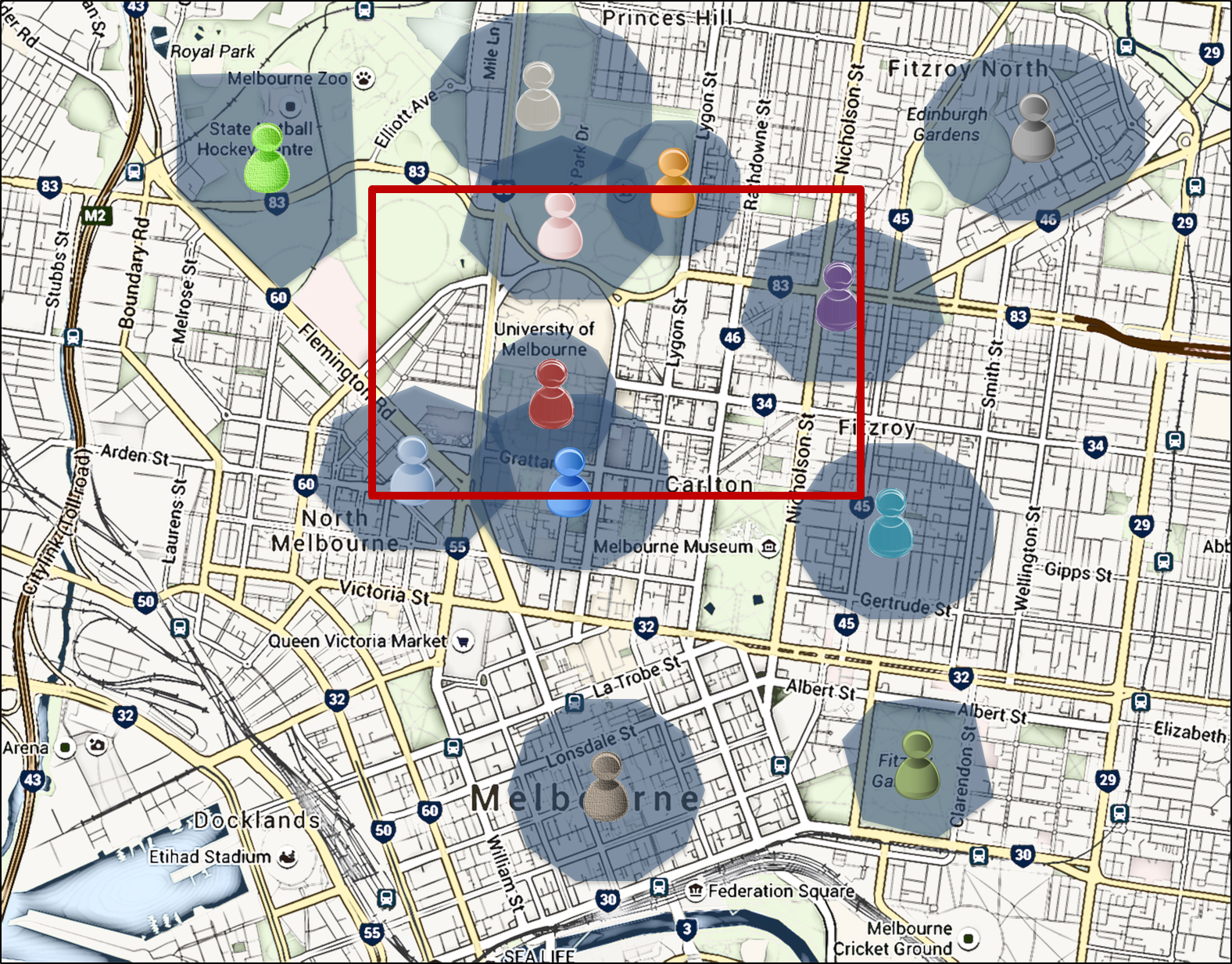}}
	\caption{Example users' point locations (path) or spatial regions on a map of Melbourne. Bolded rectangle depicts an example range query to count the number of users. }
	\label{fig:map_bodies_points}
\end{figure*}

The recently emerged strong guarantee of differential privacy~\cite{Cynthia2006,Dwork08} has attracted a number of researchers in location privacy. Typically work studies aggregation of point and trajectory data~\cite{InanKGB10,ChenFDS12,Cormode12PSD,QardajiYL13,HeCMPS15}, often via histogram-like data structures---regular or hierarchical---for controlling the level of perturbation required for privacy. 

Our goal in this paper\footnote{This paper extends our ICDM'2016 conference paper~\cite{Bodies2016}.} is to address the accurate counting of planar bodies, while providing the strong guarantee of differential privacy. 
While Euler histograms provide an excellent starting point in terms of utility, computational efficiency and aggregation-based qualitative privacy, a service provider may be directed by users to provide strong \emph{semantic} privacy. Differential privacy guarantees that an attacker with significant prior knowledge and computational resources cannot determine presence or absence of a user in a set of planar bodies. 

Differential privacy requires randomization. The challenge in combining the ideas of Euler histograms and differential privacy is that the data structure's large number of counts require randomised perturbation. As a result, the total noise added could be prohibitively high. Compared to point data in which at most one cell is impacted per record, here an object could span more than one cell, impacting many counts. Naive solutions would therefore significantly degrade utility. Moreover when sampled independently, perturbations can destroy the \term{consistency} of query responses over the resulting structure~\cite{Barak2007}.

The first stage of our approach is to perturb counts of a Euler histogram by applying noise controlled via sensitivity to a natural bound on planar body size. 
Then, to re-instate consistency and improve utility with no cost to privacy, we apply constrained inference that seeks to minimally update counts to satisfy consistency constraints. These constraints reflect relationships between data structure counts that must exist, but may be violated by perturbation. Under these constraints we apply least absolute deviations (LAD), which is more robust to outliers than ordinal regression---used previously for constrained inference in differential privacy. By enforcing consistency, we also ``average out'' previously-added noise, thereby improving utility in certain cases.
Finally, we round counts so that query responses are integral. This final stage, combined with consistency, yields responses that preserve a covertness property such that third-party observers cannot determine that privacy-preserving perturbation has taken place. 

Two privacy models have been studied for releasing datasets or their statistics: the interactive and non-interactive
models~\cite{Cynthia2006}. In the non-interactive setting, the database is sanitized and then released while the
interactive model considers mechanisms that respond to queries by releasing approximate query
responses. The main limitation with this latter approach is the limited number of queries permitted throughout the mechanism's lifetime. Interactive mechanisms (\eg Euler histograms~\cite{CASE2015}) can provide inconsistent results also. Our focus is on the non-interactive privacy setting, wherein our mechanisms release privacy-preserving data structures to third parties, with no limitation on the number of subsequent query responses permitted.
\newpage 
\paragraph{Contributions.}
We deliver several main contributions:
\begin{itemize}
	\item For the first time, we address the differentially-private counting of planar bodies (spatial region objects) in the non-interactive setting;
	\item We propose differentially-private mechanisms that leverage the Euler characteristic (via the Euler histogram data structure) to address the duplicate counting problem;
	\item We formulate novel constrained inference to reduce noise and introduce consistency based on the robust method of least absolute deviations; combined with rounding, this guarantees a covertness property;
	\item We contribute an end-to-end theoretical analysis of both high-probability utility and differential privacy; and
	\item We conduct a comprehensive experimental study on real-world datasets, which confirms the suitability of our approach to private range queries on spatial bodies.
\end{itemize}

\section{Related Work}\label{sec:related_work}

A series of effective privacy attacks on location data~\cite{Gruteser2004,Krumm07,Ghinita2013} has launched a significant amount of activity around privacy-preserving techniques for spatial analytics~\cite{Chow2011,Ghinita2013,Krumm09}.

Aggregation under range queries has emerged as a fundamental primitive in spatial and spatio-temporal analytics~\cite{Papadias2001,Tao2004,Lopez2005,Braz2007,LeonardiORRSAA14}. Originally motivated by statistical and computational efficiency, aggregation is now also used for qualitative privacy. 

A key challenge in aggregation is the \term{distinct counting}~\cite{Papadias2001,Tao2004,Lopez2005,Braz2007,LeonardiORRSAA14} or \term{multiple-counting} problem~\cite{SunAA02,Sun2006}. In contrast to point objects, a spatial body can span more than one cell in a partitioned space, inhibiting the ability of regular histograms to form accurate counts. \term{Euler histograms}~\cite{Beigel1998} are designed to address this problem for convex bodies~\cite{SunAA02,Sun2006}, by appealing to Euler's formula from graph theory~\cite{trudeau1993introduction}. A variation of Euler histogram has been studied for trajectory data
to address aggregate queries on moving objects~\cite{Xie2007}. In that work, Euler histograms were used in a distributed setting (motivating a distributed
Euler histogram), to tackle the duplicate (distinct) entry problem rather than duplicate (distinct) counting. The Euler-histogram tree~\cite{Xie2014} has been studied as a tree-based data structure for counting vehicle trajectories using the approach first developed in~\cite{CASE2015} to address the distinct counting problem for reducing storage requirements. 

There is a line of work~\cite{CASE2015}, in which the CASE histogram has been proposed as a privacy-preserving approach for trajectory data analytics, where only count data is utilised in a partitioned space applying the Euler characteristic to address duplicate counting. The authors in~\cite{CASE2015} discuss the interactive setting for differentially-private Euler histogram release, which has a prohibitive limitation of the number of queries being linear in the number of bodies. Our work has no such limitation~(see~\cite{Dwork08}).

Differential privacy~\cite{Cynthia2006,Dwork08} has now become a preferred approach to data sanitisation as it provides a strong semantic guarantee with minimal assumptions placed on the adversary's knowledge or capabilities.
Differential privacy has been studied for location privacy~\cite{Ghinita2013}. One existing approach is to obfuscate the user's location by perturbing their real geographic coordinates. The concept of geo-indistinguishability has been defined~\cite{AndresBCP13,PrimaultMLB14} as a notion of differential privacy in location-based services.
Due to its popularity, differential privacy has been applied to many algorithms and across many domains, such as specialized versions of spatial data indexing
structures designed with differential privacy for the purpose of private record matching~\cite{InanKGB10}; in spatial crowdsourcing to help volunteer workers' locations remain private~\cite{ToGS14}; in machine learning, releasing differentially-private learned models of SVM classifiers~\cite{Rubinstein12}; in geo-social networks for location recommendation~\cite{ZhangGC14};
and for modelling human mobility from real-world cellular network data~\cite{Mir13}.

Within the scope of aggregation, studies in the area of point privacy have also proposed sanitization algorithms for generating differentially-private histograms and releasing aggregate statistics. Many studies have explored differential privacy of point sets~\cite{AcsCC12,InanKGB10,Cormode12PSD,ChenFDS12,WangZM13,FanXS13,HeCMPS15,QardajiYL13,LiHay14}. They have studied regular grid partitioning data structures and hierarchical structures. This work for the first time addresses the problem of differentially-private counting of planar bodies.

Table~\ref{table:comparison}, demonstrates various techniques for privacy preserving spatial analytics using aggregates comparing privacy model, data type and approach.

\begin{table}
	\centering
	\caption{Taxonomy on private spatial data analytics using aggregates with examples of related work.}\label{table:comparison}
	\bigskip
	\scalebox{.95}{
	\begin{tabular}{|p{2.3cm}|p{1.7cm}|p{7cm}|}
		\hline
		\textbf{Privacy Model} & \textbf{Data Type} & \textbf{Approach}\\
		\hline
		\hline
		Spatial Aggregation & Trajectory & Probabilistic counting using sketches---approximation method (Tao \etal~\cite{Tao2004})
		
		Distributed Euler Histograms (DEHs) addressing \dep (Xie \etal~\cite{Xie2007}), Count-based approach similar to~\cite{Xie2007} (Leonardi \etal~\cite{LeonardiORRSAA14})
		
		CASE histograms, addressing \dop(duplicate counting) \cite{CASE2015}\\		
		\hline
		Differential Privacy & Point & Quad-tree, KD-tree (Cormode \etal~\cite{Cormode12PSD}), Uniform and Adaptive Grid (Qardaji \etal~\cite{QardajiYL13})\\
		\hline
		Differential Privacy & Trajectory & Prefix tree (Chen \etal~\cite{ChenFDS12}), DPT, using hierarchical reference systems (He \etal~\cite{HeCMPS15}), CASE Histograms~\cite{CASE2015}\\
		\hline
		Differential Privacy & Spatial Region (Planar Body) & Differentially private Euler histograms (this work)\\
		\hline
	\end{tabular}}
\end{table}

\section{Preliminaries}\label{sec:definition}

One natural but qualitative approach to privacy preservation is spatial aggregation. We will leverage a data structure that permits spatial aggregation for body counts. 
\subsection{Euler Histograms}
\label{sec:prelim-euler}

Given a grid partitioned space, an Euler histogram data structure allocates buckets not only for grid cells, but also for grid cell edges and vertices. We formally define the data structure as below.

\begin{definition}\label{def:euler}
	Consider an arbitrary partition of a subset of $\mathbb{R}^2$ into convex cells. Define \Fset, \Eset, \Vset to be index sets over the partition's faces, edges (face intersections), and vertices (edge intersections). Let $\mathbf{P}$ be a vector with components, the faces, edges and vertices, indexed by $\Fset\cup\Eset\cup\Vset$ (\ie each $P_i\subset\mathbb{R}^2$ represents a face/edge/vertex area of the Euclidean plane); and let vector $\mathbf{H}$ of non-negative integers be indexed by $\Fset\cup\Eset\cup\Vset$ as well (representing counts per face/edge/vertex). Then we call the data structure $(\mathbf{P},\mathbf{H},\Fset,\Eset,\Vset)$ an \term{Euler histogram}.
\end{definition}

Originally, Euler histograms were designed as a grid partitioning data structure, but they are valid for other convex partitions as well. For example, valid Euler histograms could be defined over a Voronoi partition of space induced by a finite set of sensors as the sites of a Voronoi diagram detecting any object in their region~\cite{Xie2007}; or a rectangular partition over an urban area~\cite{CASE2015} such as in Figure~\ref{fig:euler_histograms}.

\begin{figure}
	\centering
	\includegraphics[width =  \columnwidth]{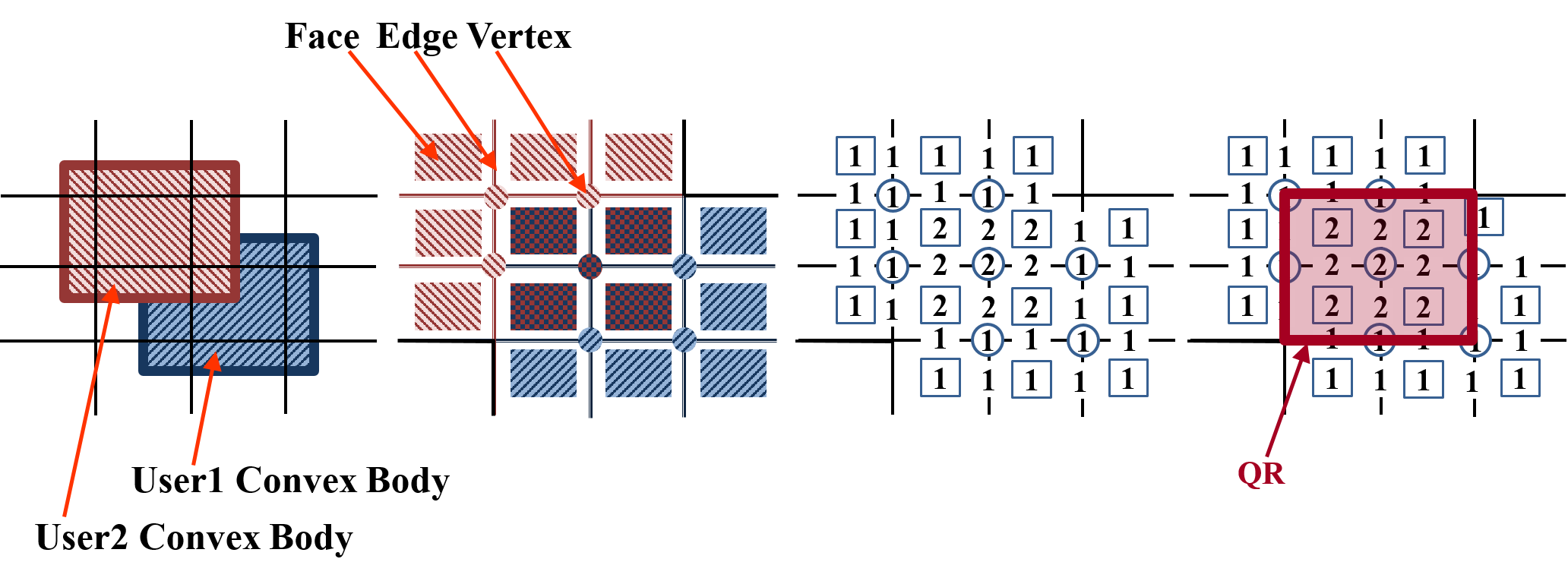}
	\caption{Two convex bodies overlapping a spatial partition and their related counts to corresponding Euler histogram; an example query region (QR) to count the number of objects.} 
	\label{fig:euler_histograms}
\end{figure}

Beigel and Tanin~\cite{Beigel1998} first introduced to spatial databases, the observation that the Euler characteristic~\cite{trudeau1993introduction} (including its extensions to higher dimensions) directly applies to this data structure. Euler's characteristic states that the number of convex bodies $N$ overlapping certain query regions can be computed exactly as
\begin{equation}\label{eq:euler}
N = F - E + V\enspace,
\end{equation}
where $F,E,V$ are the sum of face, edge, and vertex counts in $\mathbf{H}$ within the given query region (\term{QR} in Figure~\ref{fig:euler_histograms}). Duplicate counting due to summing face counts is corrected by subtracting edge counts. This in turn can over-compensate, and is corrected by adding vertex counts. This is a special case of the Inclusion-Exclusion Principle of set theory and applied probability. Figure~\ref{fig:euler_histograms} illustrates the impact two planar bodies have on a square-partition Euler histogram. Compared to conventional histograms, with the use of extra counts for grid cell edges and vertices, large objects spanning more than one cell are now distinguishable from several small objects intersecting only one cell. Applying Equation~\eqref{eq:euler} to calculate the number of objects inside the highlighted \term{QR} of Figure~\ref{fig:euler_histograms}, we arrive at the correct answer of $N = 8 - 8 + 2 = 2$.

\subsection{Differential Privacy}

We consider statistical databases on records---each representing a user's spatial region. Randomisation is vital for preventing an adversary from inverting a released statistic to reconstruct the original (private) data.

\begin{definition} 
	A \term{randomised mechanism} $\M(D)$ on database $D$, is a random variable taking values in response set $Range(\M)$.
\end{definition}

\begin{definition}We say that two databases $ D $, $ D' $ are \term{neighbours} if they are of equal size and differ on exactly one record---one spatial body representing a user in the context of this paper.
\end{definition}

\begin{definition}\label{def:dp}
	A randomised mechanism $\M$, is said to preserve \term{$\epsilon$-differential privacy} for $\epsilon>0$, if for all neighbouring databases $ D $, $ D' $, which differ in exactly one record, and measurable $C \subseteq Range(\M)$:
	\begin{eqnarray*}
		\Prob{\M(D)\in C} &\leq& \exp(\epsilon) \cdot \Prob{M(D')\in C} \enspace.
	\end{eqnarray*}
\end{definition}

Definition~\ref{def:dp} implies that an algorithm is differentially private if a change, addition or deletion of a record, does not significantly affect the output distribution. 
Differential privacy has become a \emph{de facto} standard for privacy of input data to statistical databases due to it being a semantic guarantee~\cite{Cynthia2006}.

\begin{remark}
	The threat model of differential privacy involves an incredibly powerful adversary with full knowledge of mechanism $\M$, all but one record of true latent database $D$, the ability to sample from $\M(D)$, and unlimited computational power. Using these capabilities, an optimal attack for reconstructing $D$ is to sample $m_1,\ldots,m_k\stackrel{iid}{\sim} \M(D)$. From this sample the attacker can form a histogram that is an empirical estimate $\hat \M(D)$ of the true response distribution $\M(D)$. Knowing that the true $D$ is neighbouring to the database $D'$ known by the attacker, they may simulate (using their unbounded computational resources) each and every response distribution $\M(D'')$ for neighbouring $D'',D'$ and then attempt to match these against $\hat M(D)$. Differential privacy states exactly that each of the simulated candidate response distributions are exceedingly similar (multiplicatively pointwise close), and so for sufficiently small $\epsilon$ (relative to sample size $k$ which is limited to linear in size of $D$) it is impossible for the attacker to distinguish the true $\M(D)$ by comparison with $\hat \M(D)$.
\end{remark}

\begin{lemma}[Post-Processing Immunity \cite{Cynthia2006}]\label{lem:post-immunity}
	For any randomised algorithm $\M : \mathcal{X} \rightarrow \mathcal{R}$ and any (possibly randomised) function $f: \mathcal{R} \rightarrow \mathcal{R'}$, if $\M$ is $\epsilon$-differentially private then $f\circ\M$ is also $\epsilon$-differentially private.
\end{lemma}

Lemma~\ref{lem:post-immunity} implies that differential privacy is immune to post-processing. This is also referred as \term{Transformation Invariance}, as one of the privacy axioms~\cite{Kifer2010}, indicating that post-processing privatised data maintains privacy.

\section{Problem Statement}\label{sec:problem}

The focus of this paper is to respond to range queries over spatial datasets consisting of a spatial region per user.

\begin{problem}\label{prob:main}
	Given a set of planar bodies, our goal is to batch process them to produce a data structure that can respond to an unlimited number of range queries within some fixed, bounded area: given a query region \term{QR}, we are to respond with an approximate count of bodies overlapping that region.
\end{problem}

For example, a range query covering the entire area in  Figure~\ref{fig:map_bodies} might elicit a response of (exact count of) 12.

\subsection{Evaluation Metrics}

We consider four properties of mechanisms, as competing metrics for evaluating solutions to Problem~\ref{prob:main}.

\begin{enumerate}[P1.]
	\item \textbf{Utility}: We measure utility by the absolute error of query responses relative to the true count of bodies intersecting a given query region.
	\item \textbf{Privacy}: Mechanisms should achieve non-interactive differential privacy, at some level $\epsilon$, in their release of a data structure on sensitive spatial data. 
	\item \textbf{Consistency}: \label{def:consistency}
	If responses to all possible queries agree with some fixed set of bodies then we say that the mechanism is \term{consistent}. Such a set of bodies need not coincide with the original input bodies.
	\item \textbf{Covertness}: \label{def:covert}
	If a consistent counting mechanism's query responses are non-negative integer-valued, then we also call it \term{covert}.
\end{enumerate}

Utility and privacy are in direct tension, for establishing privacy typically involves reducing the influence of data on responses. However for fixed levels of privacy, for example, we can ask what levels of utility are possible for available solutions to Problem~\ref{prob:main}.

If privacy-preserving perturbations are made independently across a data structure, it is unsurprising that overlapping queries will not necessarily result in consistent responses. This may be undesirable for some applications that utilise multiple, overlapping queries \eg urban planning. 
We consider specific, public consistency constraints which relate to the data structure adopted. As such, the \emph{level} of consistency can be benchmarked according to the number of consistency violations suffered. Unlike privacy, consistency is not necessarily at odds with utility: indeed we will demonstrate how imposing consistency can actually improve utility. Intuitively, if privacy-preservation involves injecting independent, random perturbations to a data structure, then consistency corresponds to a public smoothness assumption that can be used to `cancel out' the deleterious effect of perturbation. 
Consistency may also be applied when a measure of `stealth' is desired for a counting mechanism. 

\subsection{Assumptions}\label{sec:assumptions}

The theoretical guarantees developed in this paper leverage four assumptions (\cf Figure~\ref{fig:grid}). Each is relatively weak, being well motivated and satisfied in most practical settings.

\begin{figure}
	\centering
	\includegraphics[width = 0.4\columnwidth]{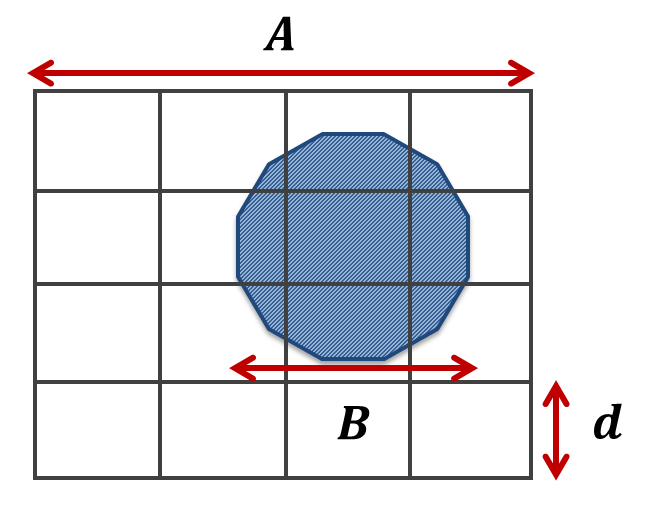}
	\caption{A convex body with bounded diameter, on a spatial partition.}
	\label{fig:grid}
\end{figure}

\begin{enumerate}[{A}1.]
	\item \label{assume:convex_cell}
	We assume that the space partition's cells are all convex.
	
	\item \label{assume:convex_qr}
	We assume that query regions are convex unions of our space partition's cells.
	
	\item \label{assume:convex_body}
	We assume that all planar bodies are convex.
	
	\item \label{assume:bound}
	We assume that all planar bodies are of some bounded $L_2$ diameter $B>0$.
	
\end{enumerate}

A sufficient condition for correctness of Equation~\eqref{eq:euler}, is that all objects are convex planar bodies. However, convexity is not necessary. In general, objects being disconnected leads to inaccurate counts. Note that connected objects can become disconnected \eg a concave object not contained by a query region~\cite{CASE2015}. Our first three assumptions are sufficient for guaranteeing correctness (perfect utility) for Euler histograms. Relaxing these assumptions may come at the cost of utility. For example convex query regions that are not unions of cells can exactly count the number of bodies in the (enlarged) union of cells intersecting the QR. And general query regions will still result in excellent utility. Two important partition geometries satisfy these conditions: rectangular and Voronoi partitions.

The fourth assumption controls the $L_1$-Lipschitz smoothness of Euler histogram counts with respect to input bodies. This parameter---also known as the \term{global sensitivity} (\cf Definition~\ref{def:gs})---calibrates the scale of noise added for differential privacy. We consider a motivating example to be regions of frequent visitation. These are necessarily bounded. With $B$ sufficiently large, no restriction is made on valid bodies.

Without loss of generality we assume partitions are square of side length $A>0$, divided into $n$ rows and $n$ columns, yielding square cells of side length $d=A/n$ (\cf Figure~\ref{fig:grid}).

\section{Algorithms and Analysis}\label{sec:approach}

Our approach consists of four complementary algorithms:

\begin{compactitem}
	\item \Eu \term{(Eu)}: Euler histogram construction from a set of convex planar bodies;
	
	\item \DP \term{(DP)}: Calibrated perturbation of histogram counts to achieve $\epsilon$-differential privacy. To improve utility, negative counts are truncated at zero;
	
	\item \LP \term{(LP)}: Constrained inference for consistency;
	
	\item \R \term{(R)}: Rounding counts for covertness. 
\end{compactitem}

We detail each mechanism, followed by its theoretical analysis. Figure~\ref{fig:algorithm_steps} depicts an example run of each algorithm in turn. As shown, Figure~\ref{fig:bodies} illustrates two users' spatial regions, our running example in Figure~\ref{fig:euler_histograms}, as input raw data for the first Algorithm~\ref{algo:euler} (\cf Section~\ref{sec:euler}). 
\begin{figure*}
	\centering
	\subfloat[ Two users' planar bodies.\label{fig:bodies}]{
		\includegraphics[scale = .29]{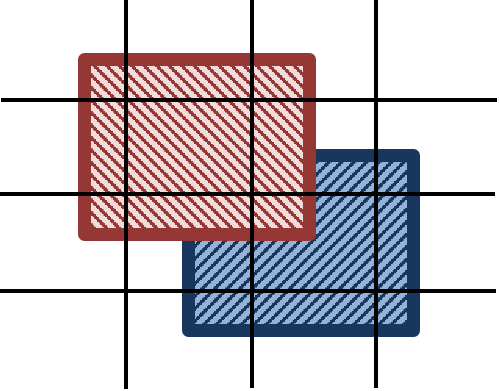}}
	\hfill
	\subfloat[ Euler histogram construction.\label{fig:Euler}]{
		\includegraphics[scale = .29]{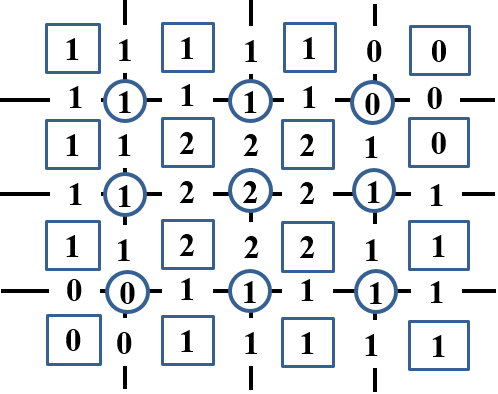}}
	\hfill
	\subfloat[ Perturbing the counts (DP).\label{fig:DP}]{
		\includegraphics[scale = .29]{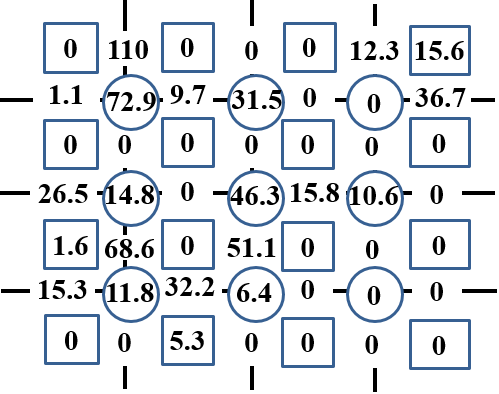}}

	\subfloat[ Linear programming (LP).\label{fig:LP}]{
		\includegraphics[scale = .29]{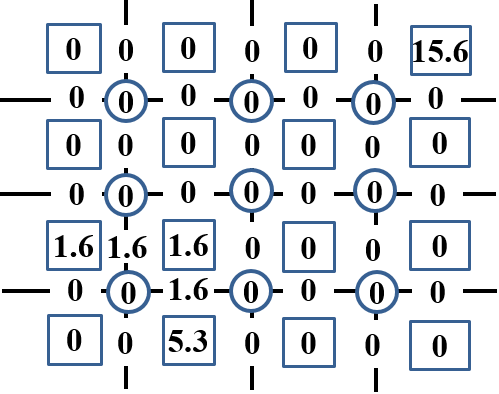}}
	\hfill
	\subfloat[ Rounding the counts (R).\label{fig:round}]{
		\includegraphics[scale = .29]{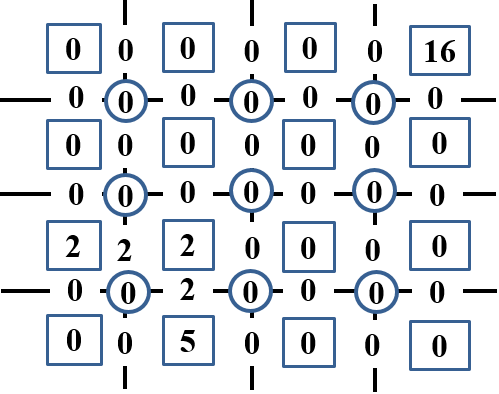}}
	\hfill
	\subfloat[ Example range query response.\label{fig:QR}]{
		\includegraphics[scale = .29]{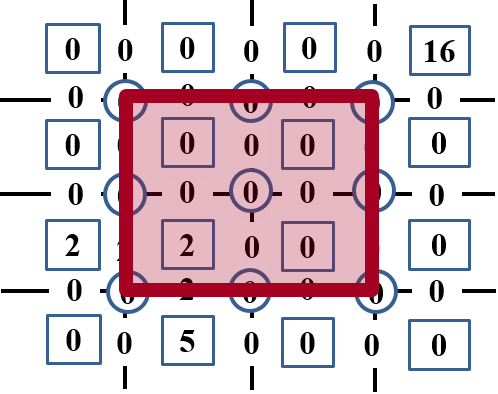}}
	\caption{An example of the mechanisms' outputs; numbers from a real run. }
	\label{fig:algorithm_steps}
\end{figure*}

\subsection{Algorithm: Euler}\label{sec:euler}

Algorithm~\ref{algo:euler} creates a data structure (Euler histograms \cf Section~\ref{sec:prelim-euler}) to represent aggregated counts of a given set of convex planar bodies $\mathcal{X}$. The algorithm simply increments counts for any face, edge, vertex that intersects a body. As shown in Figure~\ref{fig:Euler}, processing a convex body determines what counts need to be incremented.

\IncMargin{1em}

\begin{algorithm}
	\LinesNumbered
	\SetKwInOut{Input}{Input}\SetKwInOut{Output}{Output}
	
	\Input{Set of planar bodies $\mathcal{X}$; partition $(\mathbf{P},\Fset,\Eset,\Vset)$}
	\Output{Euler histogram $(\mathbf{H},\mathbf{P},\Fset,\Eset,\Vset)$}
	\For{$i \in \Fset \cup \Eset  \cup \Vset$}{
		$\h{i} \longleftarrow 0$
	}
	\For{$x \in \mathcal{X} $}{
		\For{$i \in \Fset \cup \Eset  \cup \Vset$}{
			\If{  $x \cap P_i \neq \emptyset$}{
				$\h{i} \longleftarrow \h{i} + 1$
			}
		}
	}

	\caption{Euler (Eu): Euler Histogram Construction}\label{algo:euler}
	
\end{algorithm}
\DecMargin{1em}

\paragraph{Privacy.}
\Eu is qualitatively private via aggregation, 
but it does not achieve any differential privacy by virtue of being deterministic.

\paragraph{Utility.}
Assumptions A\ref{assume:convex_cell}--A\ref{assume:convex_body} guarantee the preconditions of the following, direct results of Equation~\eqref{eq:euler}. 

\begin{corollary}
	If input bodies, partition cells, and query region are convex, and the query region is a union of cells, then \Eu's responses to the range query via Equation~\eqref{eq:euler} are accurate.
\end{corollary}

\begin{corollary}
	\Eu is consistent (P\ref{def:consistency}) and covert (P\ref{def:covert}).
\end{corollary}

\paragraph{Computational Complexity.} As our partition has $n$ rows and columns, \Eu's time and space complexities are efficient at $O( \abs{\mathcal{X}} n^2)$ and $O(n^2)$ respectively.

\subsection{Algorithm: DiffPriv}\label{sec:DP}
\Eu achieves a number of our target properties but not differential privacy. We now introduce differential privacy to our approach by perturbing Euler histogram counts. In Algorithm~\ref{algo:DP}, we add carefully-crafted random noise 
based on the sensitivity of the histogram to input bodies. 
We truncate any resulting negative counts to zero, improving utility at no cost to privacy (\cf Lemma~\ref{lem:post-immunity}). Figure~\ref{fig:DP} depicts the result of this phase for a real example. For interested readers, the computed global sensitivity (GS) (\cf Definition~\ref{def:gs} and Lemma~\ref{lem:gs}) of this example is $25$, where $\ceil{B/d} = 2$, and the allocated privacy budget, $\epsilon$, is $1$.

\paragraph{Privacy.} The key step to establishing the differential privacy of \DP, is to calculate Lipschitz smoothness for \Eu---the scale of noise to be added to reduce sensitivity. This represents how sensitive \Eu is to input bodies, and so how much noise should be added to reduce this sensitivity for privacy.

\begin{definition}\label{def:gs}
	Let $f$ be a deterministic, real-vector-valued function of a database. The \term{$L_1$-global sensitivity} (GS) of $ f $ is given by $\Delta f = \max\limits_{D,D'} \norm{f(D) -f(D')}_1$, taken over all neighbouring pairs of databases. 
\end{definition}

The $L_1$-global sensitivity is a property of function $f$, independent of input database. For Euler histograms, the GS measures the effect on the histogram count vector, due to changing an input planar body related to a user's spatial region. 

\begin{lemma}\label{lem:gs}
	The $L_1$-global sensitivity of \Eu is $4.5\Big(\ceil*{\frac{B}{d}}+1\Big) \ceil*{\frac{B}{d}}$, where $d>0$ is the cell side length, and $B>0$ is an $L_2$ bound on planar body diameter.
\end{lemma}

\begin{proof}
	By Assumption~\ref{assume:bound} (\cf Figure~\ref{fig:grid}), the number of cells that could intersect with a body is at most $\ceil*{\frac{B}{d}} + 1$ in one direction. Therefore the total number of cells that could intersect a body is
	\begin{eqnarray*}
		n^2 &\le& \left(\ceil*{\frac{B}{d}}+1\right)^2\enspace.
	\end{eqnarray*}
	
	From this the number of faces, edges and vertices of partition $\mathbf{P}$ intersecting with a body can be upper-bounded as
	\begin{eqnarray*}
		\mbox{\#Faces} &=& n^2  \leq \left(\ceil*{\frac{B}{d}}+1\right)^2 \enspace; \\
		\mbox{\#Edges} &\leq& 2 n (n-1) \enspace;\enspace \mbox{and} \\
		\mbox{\#Vertices} &\leq& (n - 1)^2 \enspace.
	\end{eqnarray*}
	
	Summing these, we may bound the total number of partition components intersected by the body as
	\begin{eqnarray*}
		&&		 4n(n-1) + 1 \\
		& \leq& 4\left(\ceil*{\frac{B}{d}}+1\right) \ceil*{\frac{B}{d}} + 1 \\
		& =& 4\left(\ceil*{\frac{B}{d}}+1\right) \ceil*{\frac{B}{d}} + \frac{1}{2} \ceil*{\frac{B}{d}}^2 +  \frac{1}{2} \ceil*{\frac{B}{d}} \\
		&= &4\left(\ceil*{\frac{B}{d}}+1\right) \ceil*{\frac{B}{d}} + \frac{1}{2} \left(\ceil*{\frac{B}{d}}+1\right) \ceil*{\frac{B}{d}}\\
		& \leq&  4.5\left(\ceil*{\frac{B}{d}}+1\right) \ceil*{\frac{B}{d}}\enspace.
	\end{eqnarray*}
	Since changing a single body in a database can affect impacted histogram cell counts by one, this expression is also a bound on global sensitivity.
\end{proof}

\DP applies the \term{Laplace mechanism}~\cite{Cynthia2006} to \Eu: it adds to a non-private vector-valued function $f$, i.i.d. Laplace-distributed noise with centre zero and scale $\lambda$ given by $\Delta f /\epsilon$, for desired privacy level $\epsilon>0$. Here, $\lambda = \Delta \Hb/\epsilon$.

\IncMargin{1em}
\begin{algorithm}[t]
	\LinesNumbered
	\SetKwInOut{Input}{Input}\SetKwInOut{Output}{Output}
	
	\Input{Euler histogram: $(\mathbf{P},\mathbf{H}, \Fset, \Eset, \Vset)$; privacy  $\epsilon>0$; sensitivity $\Delta\mathbf{H}>0$}
	\Output{Noisy histogram: $(\mathbf{P},\mathbf{H}', \Fset, \Eset, \Vset)$}
	
	\For{$i \in \Fset \cup \Eset  \cup \Vset$}{
		$\h{i}' \longleftarrow \h{i} + Lap(0;\Delta\mathbf{H}/\epsilon)$
		
		\If { $ \h{i}' < 0 $ }
		{ $ \h{i}'  \longleftarrow 0 $}
	}
	
	\caption{DiffPriv (DP): Perturbation by Laplace Noise}\label{algo:DP}
	
\end{algorithm}
\DecMargin{1em}

\begin{corollary}
	\DP preserves $\epsilon$-differential privacy.
\end{corollary}

\begin{proof}
	The result follows by applying the triangle inequality to the odds ratio using the definition of Laplace density, and global sensitivity~\cite{Cynthia2006}. 
\end{proof}

\paragraph{Utility. } \DP is neither covert nor consistent, however we can bound its utility.

\begin{theorem}\label{theorem:DP}
	For confidence level $\delta \in (0,1)$, the counts $\Hb$ output by \Eu and counts $\Hn$ output by \DP are uniformly close with high probability
	\begin{eqnarray*}
		\Pr\left(\norm{\Hn - \Hb}_\infty \leq \lambda \log\left(\frac {\abs{\Fset} + \abs{\Eset} + \abs{\Vset}}{\delta}\right)\right) &\geq& 1 -\delta\enspace.
	\end{eqnarray*}
\end{theorem}

\begin{proof}
	For convenience, we define the combined index set $\Hset = \Fset \cup \Eset \cup \Vset$, noting that
	$\abs{\Hset} = \abs{\Fset} + \abs{\Eset} + \abs{\Vset}$. Recall that by the definition of \DP, we have that
	\begin{eqnarray*}
		\forall i \in \Hset, && \hn{i} = \h{i} + Y_{i}, \enspace Y_{i} \sim Lap(0; \lambda)\enspace.
	\end{eqnarray*}
	By the cumulative distribution function of the zero-mean Laplace, it follows that
	\begin{eqnarray*}
		\forall~ i \in \Hset, && \Prob{\abs{Y_i}\ge z} = \exp\left(\frac{-z}{\lambda}\right)\enspace,
	\end{eqnarray*}
	for any scalar $z>0$. By the union bound it follows that
	\begin{eqnarray*}
		\mathrm{Pr}\left(\bigcup \limits_{i \in \Hset} \{\abs{Y_i} \ge z\}\right) &\le& \sum_{i \in \Hset} \Prob{\abs{Y_i} \geq z} \\
		&=& \abs{\Hset} \times \exp\left(\frac{-z}{\lambda}\right)\enspace.
	\end{eqnarray*}
	Applying De Morgan's law,
	\begin{eqnarray*}
		\mathrm{Prob}\left(\bigcap \limits_{i \in \Hset} \{\abs{Y_i} < z\}\right) & =& 1 - \mathrm{Prob}\left(\bigcup \limits_{i \in \Hset} \{\abs{Y_i} \ge z\}\right) \\
		&\geq& 1- \abs{\Hset} \times \exp\left(\frac{-z}{\lambda}\right)\\  
		& \triangleq& 1 - \delta \enspace.
	\end{eqnarray*}
	Solving yields 
	\begin{eqnarray*}
		z &=& \lambda \times \log\left(\frac {\abs{\Hset}}{\delta}\right)\enspace,
	\end{eqnarray*}
	so that
	\begin{eqnarray*}
		\mathrm{Prob}\left(\bigcap \limits_{i \in \Hset} \left\{ \abs{Y_i} < \lambda \log\left(\frac {\abs{\Hset}}{\delta}\right)\right\}\right) &\geq& 1 - \delta\enspace.
	\end{eqnarray*}
	The result follows from $\Hn - \Hb = \mathbf{Y}, \enspace \mathbf{Y} \sim Lap(\lambda)$ iid. 
\end{proof}

\paragraph{Computational Complexity.} As our partition has $n$ rows and columns, \DP's time/space complexities are efficient at $O( n^2)$.

\subsection{Algorithm: Linear Programming}\label{sec:linprog}
After additive randomised perturbation with \DP, we apply constrained inference to smooth this noise, as detailed below. We begin by defining constrained inference, followed by a set of public consistency constraints.

\subsubsection{Constrained Inference: LAD}\label{sec:lp-lad}
Constrained inference models the noisy counts output by \DP as noisy observation of latent counts which are themselves related according to a set of constraints. Inference effectively smooths the differentially-private release, potentially improving utility without affecting privacy. Previously ordinary least squares (OLS) has driven constrained inference~\cite{Cormode12PSD,HayRMS10}. Here we propose instead to use least absolute deviation (LAD) (also referred to as least absolute residuals, least absolute errors and least absolute value)~\cite{Dielman2005}. In contrast to OLS, LAD has the benefit of being robust to outliers. LAD is ideal for our setting, since its choice of minimising $L_1$ error corresponds to maximising the exponential of the negative $L_1$: a Laplace noise model, akin to maximum-likelihood estimation, matching \DP precisely.

\begin{definition} Let \Hb be the Euler histogram counts with a set of defined constraints, $\mathcal{C}$. Given noisy histogram counts, \Hn, constrained LAD inference returns vector \Hc, that satisfies the constraints $\mathcal{C}$ while minimising $\norm{\Hc-\Hn}_1 $.
\end{definition}

\begin{proposition}\label{prop:lad-mle}
	Suppose Alice (\alice) wished to communicate to Bob (\bob) her parameter vector $\pmb{\theta}\in\Theta$ some Euclidean parameter family known to \bob, but that her communication of $\pmb{\theta}$ passed through a noisy channel specified by the Laplace mechanism: \bob observes $\pmb{\theta}$ with additive i.i.d. zero-mean Laplace with known scale $\lambda>0$. Then LAD corresponds to \bob using maximum-likelihood estimation to recover $\pmb{\theta}$.
\end{proposition}

\begin{proof}
	In this abstract setting (that applies beyond our mechanisms, to the Laplace mechanism more generally) suppose that \bob observes via the channel from \alice
	$$X_i\stackrel{indep.}{\sim} Lap(\theta_i, \lambda)\enspace, \enspace i\in\{1,\ldots,m\}\enspace.$$
	Then the joint likelihood of the $X_i$, known to \bob, is given by
	$$\prod_{i=1}^m \frac{1}{2\lambda}\exp\left(-\frac{|x_i-\theta_i|}{\lambda}\right)=\frac{1}{2^m\lambda^m}\exp\left(-\frac{\|\mathbf{x}-\pmb{\theta}\|_1}{\lambda}\right)\enspace.$$
	The MLE of unknown $\pmb{\theta}$ given the observations and known scale $\lambda$ corresponds to the constrained optimisation
	\begin{eqnarray*}
		\hat{\pmb{\theta}}_{MLE}
		&\in& \arg\max_{\pmb{\theta}\in\Theta}\ \frac{1}{2^m\lambda^m}\exp\left(-\frac{\|\mathbf{x}-\pmb{\theta}\|_1}{\lambda}\right) \\
		&=& \arg\max_{\pmb{\theta}\in\Theta}\ \log\left(\exp\left(-\frac{\|\mathbf{x}-\pmb{\theta}\|_1}{\lambda}\right)\right) \\
		&=&\arg\max_{\pmb{\theta}\in\Theta}\ -\frac{\|\mathbf{x}-\pmb{\theta}\|_1}{\lambda} \\
		&=& \arg\min_{\pmb{\theta}\in\Theta}\ \|\mathbf{x}-\pmb{\theta}\|_1 \enspace.
	\end{eqnarray*}
	The first equality follows from a strictly monotonic transformation of the objective function, the second follows from cancelling the logarithmic and exponential functions, and the final equality follows from another strictly monotonic transformation. This last formulation corresponds to constrained LAD.
	\end{proof}

The application of this result to our present setting involves equating the latent parameter vector to the raw histogram $\mathbf{H}$, Laplace-perturbed observations to $\mathbf{H}'$, and the parameter family $\Theta$ to constraint set $\mathcal{C}$ (to be discussed below). This connection demonstrates that our use of LAD is principled. Given public prior knowledge of counts, one could incorporate a corresponding (public) prior distribution on the $\theta$ and perform \emph{maximum a posteriori} (MAP) point-estimation which would in-turn correspond to placing a regularisation term on the LAD objective. We leave such extensions to future work.

\paragraph{Consistency.}
We define three constraints C1, C2 and C3 for Euler histograms as follows. Our consistency constraints consider the relationships between face, edge and vertex counts.  Every increment to an edge count must correspond to an increment to the counts of both incident faces as well; and similarly for an increment to a vertex count, the corresponding four incident edge counts must be incremented. Finally query regions should respond with non-zero count estimates. These represent the intuition behind our three sets of consistency constraints.

For ease of exposition, we refer to face, edge and vertex components of \Hb by $F_i, E_i, V_i$ respectively. The meaning will be apparent from context.

\begin{constraint}\label{const:C1} 
	Every edge count is less than or equal to the minimum value of its two incident faces.
	\begin{eqnarray*}
		\ec{i} \le \fc{j}\ \ \forall i \in \Eset, \forall j \in \Fset_i; \enspace \Fset_i = \{j \in \Fset: j\text{ incident to } i \in \Eset\}
	\end{eqnarray*}
\end{constraint}

\begin{constraint}\label{const:C2} 
	Every vertex count is less than or equal to its four incident edges' counts. 
	\begin{eqnarray*}
		\vc{i} \le \ec{j}\ \ \forall i \in \Vset, \forall  j \in \Eset_i; \enspace \Eset_i = \{j \in \Eset: j\text{ incident to } i \in \Vset\}
	\end{eqnarray*}
\end{constraint}
\begin{constraint}\label{const:C3}
	Every two by two grid partition should have a non-negative count computed by Euler, Equation~\eqref{eq:euler}. 
	\begin{eqnarray*}
		\fc{j} - \ec{k} + \vc{i} \ge 0  &&  \forall i \in \Vset, \forall j \in \Fset_i,  \forall  k \in \Eset_i 
	\end{eqnarray*}  
	\begin{eqnarray*}
		\mbox{where}\ \ \ \ \ 	
		\Fset_i &=& \{j \in \Fset: j\text{ incident to } i \in \Vset\}  \\
		\Eset_i &=& \{k \in \Eset: k\text{ incident to } i \in \Vset\}\enspace.
	\end{eqnarray*}
\end{constraint}

Figure~\ref{fig:LP} demonstrates the output of \LP algorithm that smooths the noise and applies consistency constraints.

\paragraph{Algorithm. }
We consider two constrained inference programs for enforcing these constraints. Both minimise the change to the histogram counts subject to the constraints. The first, LAD, minimises counts with respect to the $L_1$-norm.
\begin{eqnarray*}
	\min_{\Hc} \norm{\Hc-\Hn}_1 \ \   \mbox{s.t.} \   \Hc \geq \zero \enspace \ \  \text{Constraints } C_1, C_2, C_3
\end{eqnarray*}

By introducing a primal variable per histogram cell count, we can transform this to the following linear program
\begin{eqnarray}
\min_{\Hc, \mathbf{h}} && \sum_{i=1}^{|\Hset|}h_i \label{prog:L1} \\ 
\mbox{s.t.} && \Hc, \mathbf{h} \geq \zero \nonumber \\ 
&& \hn{i} - \hc{i} \leq h_i \enspace \forall i \in \Hset \nonumber\\ 
&& \hc{i} - \hn{i} \leq h_i \enspace \forall i \in \Hset \nonumber\\ 
&& \text{Constraints } C_1, C_2, C_3 \nonumber 
\end{eqnarray}

Alternatively we could adopt the $L_\infty$-norm for minimising the change to the histogram cell counts, as in the following program.
\begin{eqnarray*}
	\min_{\Hc}  \norm{\Hc-\Hn}_\infty  \ \ \  \mbox{s.t.} \  \Hc \geq \zero  \ \ \   \text{Constraints } C_1, C_2, C_3\\
\end{eqnarray*}

And again we may transform this program to an equivalent LP, this time by introducing only a single new primal variable
\begin{eqnarray}
\min_{\Hc, h} && h  \label{prog:Linfty} \\ 
\mbox{s.t.} && \Hc, h \geq \zero \nonumber \\ 
&& \hn{i} - \hc{i} \leq h \enspace \forall i \in \Hset \nonumber \\ 
&& \hc{i} - \hn{i} \leq h \enspace \forall i \in \Hset \nonumber \\ 
&& \text{Constraints } C_1, C_2, C_3\nonumber 
\end{eqnarray}

We analyse Program~\eqref{prog:Linfty}, however we recommend that in practice Program~\eqref{prog:L1} be used since it is better able to minimise change to all cell counts (and is derived according to the MLE principle as per Proposition~\ref{prop:lad-mle}), while Program~\eqref{prog:Linfty} only minimises the maximum error. Algorithm~\ref{algo:linprog} and our experiments reflect this recommendation.

\IncMargin{1em}
\begin{algorithm}[t]
	\LinesNumbered
	\SetKwInOut{Input}{Input}\SetKwInOut{Output}{Output}
	
	\Input{Noisy Histogram: $(\mathbf{P},\Hn,\Fset,\Eset,\Vset)$}
	\Output{Consistent Histogram: $(\mathbf{P},\Hc,\Fset,\Eset,\Vset)$}
	
	Solve Program~\eqref{prog:L1}.
	
	\caption{LinProg (LP): Linear Programming}\label{algo:linprog}
	
\end{algorithm}
\DecMargin{1em}

\paragraph{Privacy.}
Since \LP depends only on the output of \DP, it preserves the same level of differential privacy (\cf Lemma~\ref{lem:post-immunity}).

\paragraph{Utility.} We can establish high-probability utility bounds on \LP$(L_\infty)$ that take a similar form to those proved for \DP, but via different arguments.

\begin{theorem}\label{theorem:LP}
	For any confidence level $\delta \in (0,1)$, and for histogram counts \Hn output by \DP and \Hc minimising Program~\eqref{prog:Linfty}, we have
	\begin{eqnarray*}
		\mathrm{Pr}\left(\norm{\Hn - \Hc}_\infty \leq \lambda \log\left(\frac {\abs{\Fset} + \abs{\Eset} + \abs{\Vset}}{\delta}\right)\right) &\geq& 1 - \delta\enspace.
	\end{eqnarray*}
\end{theorem}

\begin{proof} We reduce to the bound on \DP, by noting that since \LP is minimising distance, the distance from \Hc to \Hn must be no more than \Hb to \Hn. In other words
	
	\begin{eqnarray*}
		\overbrace{\norm{\Hn-\Hc}_\infty}^\text{LP} &\leq \overbrace{\norm{\Hn-\Hb}_\infty}^\text{Laplace Analysis} \leq& \lambda \log\left(\frac {\abs{\Fset} + \abs{\Eset} + \abs{\Vset}}{\delta}\right) 
	\end{eqnarray*}
	with the final bound holding with probability at least $1-\delta$. 
\end{proof}

\paragraph{Computational Complexity.}
Linear programming interior-point methods---also referred to as barrier algorithms---are polynomial-time, with worst-case complexity of $O(a^{3.5})$~\cite{Karmarkar1984}, for $a$, the number of variables. Therefore, for Euler histograms the time complexity is $O(n^7)$, but in practice it is efficient as demonstrated in our runtime experiments (\cf Section~\ref{sec:timing} for running time).

\subsection{Algorithm: Rounding}
After running \LP, we introduce covertness via \R (Algorithm~\ref{algo:round}). 
This allows the data curator to hide that the data has been perturbed (see Figure\ref{fig:round}). Figure~\ref{fig:QR} depicts an example range query response to privately count the number of users via Equation~\eqref{eq:euler}, $N = F - E + V = 2 $.

\IncMargin{1em}
\begin{algorithm}[t]
	\LinesNumbered
	\SetKwInOut{Input}{Input}\SetKwInOut{Output}{Output}
	
	\Input{Consistent Histogram:  $(\mathbf{P},\Hc,\Fset,\Eset,\Vset)$}
	\Output{Rounded Histogram: $(\mathbf{P},\Hr,\Fset,\Eset,\Vset)$}
	
	\For{$i \in \Fset \cup \Eset  \cup \Vset$}{
		$\hr{i} \longleftarrow round(\hc{i})$
		
	}
	\caption{Rounding (R)}\label{algo:round}
\end{algorithm}
\DecMargin{1em}

\paragraph{Privacy.} 
Since \R depends only on differentially-private data, it also preserves differential privacy (\cf Lemma~\ref{lem:post-immunity}).

\paragraph{Utility.} The analysis of utility for \R is more straightforward than for \DP and \LP.

\begin{lemma}\label{theorem:R}
	If \Hc is the output histogram of \LP and \Hr is the result of \R, then $\norm{\Hc - \Hr}_\infty \leq 0.5$.	
\end{lemma}

\begin{lemma}
	\R is consistent when run after \LP, and so it is also covert.
\end{lemma}

\begin{proof}
	We only need to check the consistency constraints, as to whether \R violates any. This cannot happen, since the smaller side of a constraint inequality rounding up must coincide with the larger side rounding up. Similarly the larger side rounding down must coincide with the smaller side doing the same. Therefore, consistency is invariant to rounding.
\end{proof}

\paragraph{Computational Complexity.} Similar to \DP since our partition has $n$ rows and columns, \R's time and space complexities are efficient at $O( n^2)$.

\subsection{Full Theoretical Analysis}

We are now able to combine the individual utility analyses of the four stages of our approach, into an overall high-probability bound on utility.

\begin{corollary}
	For confidence level $\delta\in(0,1)$, and histogram counts \Hb, \Hr output by \Eu and \R respectively we have that
	\begin{eqnarray*}
		\norm{\Hb - \Hr}_\infty &\leq& \frac{9\big(\ceil*{\frac{B}{d}}+1\big) \ceil*{\frac{B}{d}} }{\epsilon}\log\left(\frac {\frac{4A^2}{d^2} - \frac{4A}{d} + 1}{\delta}\right)+0.5
	\end{eqnarray*}
	holds with probability at least $1 - \delta$.
\end{corollary}

\begin{proof}
	By Theorems~\ref{theorem:DP},~\ref{theorem:LP}, Lemma~\ref{theorem:R}, triangle inequality
	\begin{eqnarray*}
		\norm{\Hb-\Hr}_\infty & \leq& \norm{\Hb-\Hn}_\infty + \norm{\Hn-\Hc}_\infty + \norm{\Hc-\Hr}_\infty \\
		&\leq& 2 \times \lambda \log\left(\frac {\abs{\Fset} + \abs{\Eset} + \abs{\Vset}}{\delta}\right) + 0.5\enspace
	\end{eqnarray*}
	with high probability , where $\lambda = 4.5\left(\ceil*{\frac{B}{d}}+1\right) \ceil*{\frac{B}{d}}/\epsilon$.
	Continuing
	\begin{eqnarray*}
		&& 2 \times \lambda \log\left(\frac {\abs{\Fset} + \abs{\Eset} + \abs{\Vset}}{\delta}\right) + 0.5 \\
		& \leq& \frac{2\big[4\left(\ceil*{\frac{B}{d}}+1\right) \ceil*{\frac{B}{d}} + 1\big]}{\epsilon} \log \left(\frac {\abs{\Fset} + \abs{\Eset} + \abs{\Vset}}{\delta}\right)+0.5\\
		&\leq& \frac{9\left(\ceil*{\frac{B}{d}}+1\right) \ceil*{\frac{B}{d}}}{\epsilon}\log\left(\frac {\frac{4A^2}{d^2} - \frac{4A}{d} + 1}{\delta}\right)+0.5\enspace.\\
	\end{eqnarray*}
	We have used the following counts, where $n$ is the number of rows/columns in the grid-partitioned area of volume $A^2$:
	\begin{eqnarray*}
		\abs{\Fset} &=& n \times n = n^2 = \frac{A^2}{d^2}\enspace; \\
		\abs{\Eset} &\leq& 2n \times (n-1) = 2(n^2 - n) = 2(\abs{\Fset} - \sqrt{\abs{\Fset}})\enspace;\\
		\abs{\Vset} &\leq& (n-1)^2 = n^2 - 2n + 1 = \abs{\Fset} - 2\sqrt{\abs{\Fset}} + 1\enspace;\\
		\abs{\Fset} + \abs{\Eset} + \abs{\Vset} &\leq& \abs{\Fset} + 2(\abs{\Fset} - \sqrt{\abs{\Fset}})+ \abs{\Fset} - 2\sqrt{\abs{\Fset}} + 1 \\
		&=& 4\abs{\Fset} - 4\sqrt{\abs{\Fset}} + 1 \\
		&\leq& 4\frac{A^2}{d^2} - 4\sqrt{\frac{A^2}{d^2}} + 1 \\
		&=& \frac{4A^2}{d^2} - \frac{4A}{d} + 1\enspace.
	\end{eqnarray*}
	This completes the proof. 
\end{proof}

Note, the utility bound's error is $O\left(\frac{B^2}{\epsilon d^2}\log \left(\frac {A^2}{\delta d^2}\right)\right)$ with high probability.

\begin{remark}
	In order to achieve appropriate utility, we recommend selecting cell size $d$, based on third-party requirements. The smallest QR that a third party might run on an area is a reasonable choice for $d$. 
	$B$ can naturally be set by users or service provider. There is little risk that $B$ would be made too large, as a user cannot have a very large region representing their regular location in a short time interval. In \eg fitness applications, users can determine the area in which they usually perform their workouts. 
	Regarding the $\epsilon$ parameter, there are studies in the literature discussing how to set this parameter~\cite{Dwork11,Hsu14}. In fact, there is a trade-off between $\epsilon$ and accuracy. Ultimately these must be set depending on third-party requirements.
\end{remark}

\section{Experimental Study}\label{sec:experiment}

\subsection{Datasets}

We conduct extensive experiments on three real-world datasets that vary in terms of density and concentration of locations. One dataset records GPS coordinates of more than 500 taxis over 30 days in the San Francisco Bay Area; Cab mobility traces are provided through the Cabspotting project~\cite{comsnets09piorkowski}. Here, cabs' GPS points are more concentrated on the financial district and surrounding areas (\cf Figure~\ref{fig:san_point}); we select this area for the empirical study (\cf Figure~\ref{fig:san_zoom}).
Our remaining datasets are in Beijing (Microsoft Research Asia), Geolife project Version 1.3~\cite{ZhengXM10}, as well as T-Drive~\cite{YuanZZXXSH10}. In Geolife 1.3, GPS trajectories were collected by 182 users, containing~18,000 trajectories. 91.5 percent of the trajectories are logged in a dense representation (every 1--5 seconds or every 5--10 meters per point). GeoLife dataset gathered a broad range of users' outdoor movements, including not only everyday routines---\eg going home and commuting to work---but also entertainment and sporting activities, including shopping, sightseeing, dining, hiking, and cycling. T-Drive includes the GPS trajectories of about 10,000 taxis within Beijing, with a total number of points at about 15 million. The distribution of users' spatial bodies over the map of the selected area per dataset is visualised in Figure~\ref{fig:datasets_heatmaps}. The distribution of the GeoLife 1.3 and T-drive are different in the same selected area (\cf Figures~\ref{fig:tdrive_heatmap_mesh} and~\ref{fig:geo13_heatmap_mesh}). Compared to GeoLife, T-Drive has a more spread distribution of users' spatial regions over the partitioned space.

\begin{figure*}[b]
	\centering
	\subfloat[\small T-Drive (Beijing).\label{fig:tdrive_heatmap_mesh}]{
		\includegraphics[scale = .13]{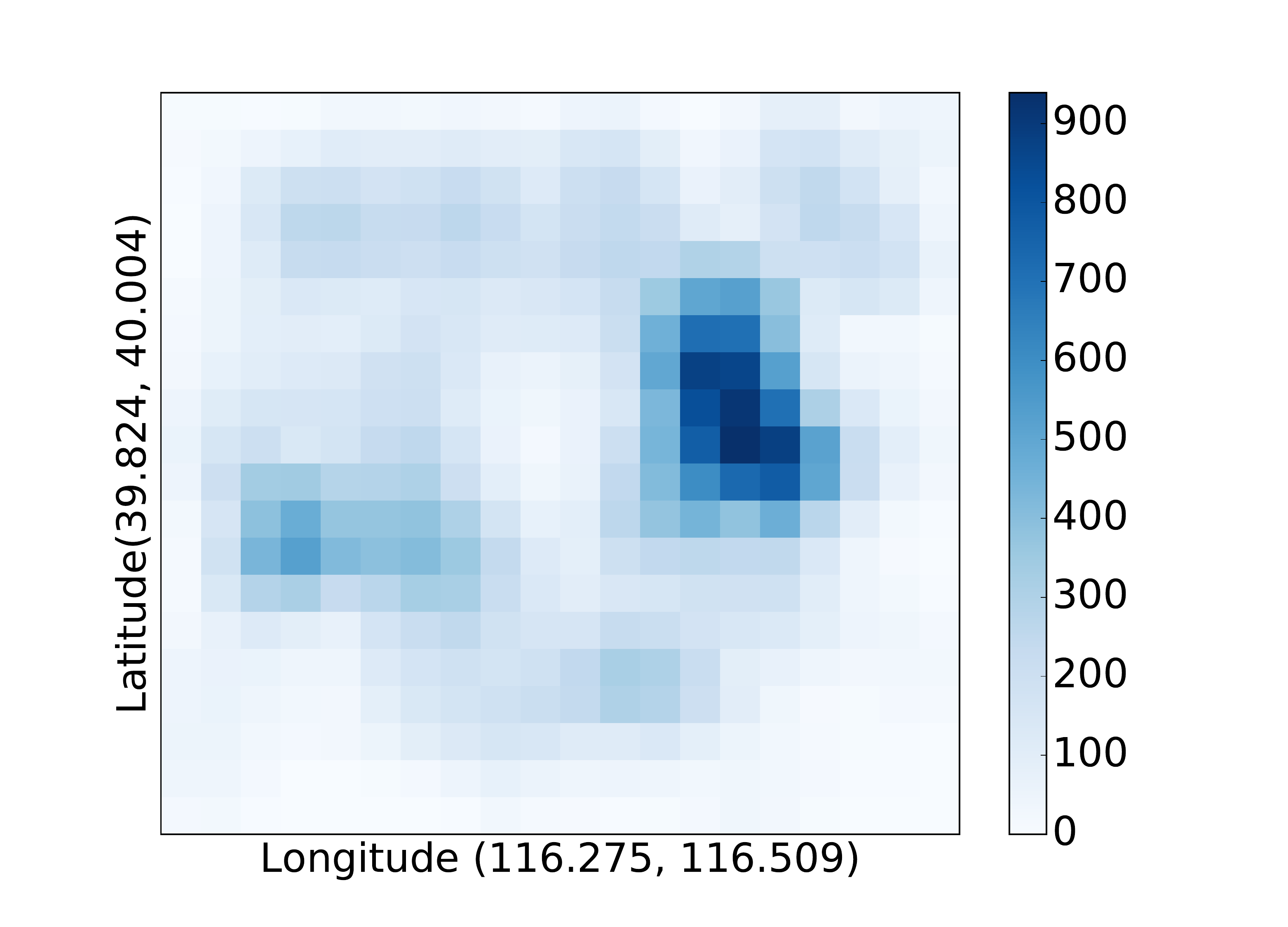}}
	\subfloat[\small GeoLife1.3 (Beijing).\label{fig:geo13_heatmap_mesh}]{
		\includegraphics[scale = .13]{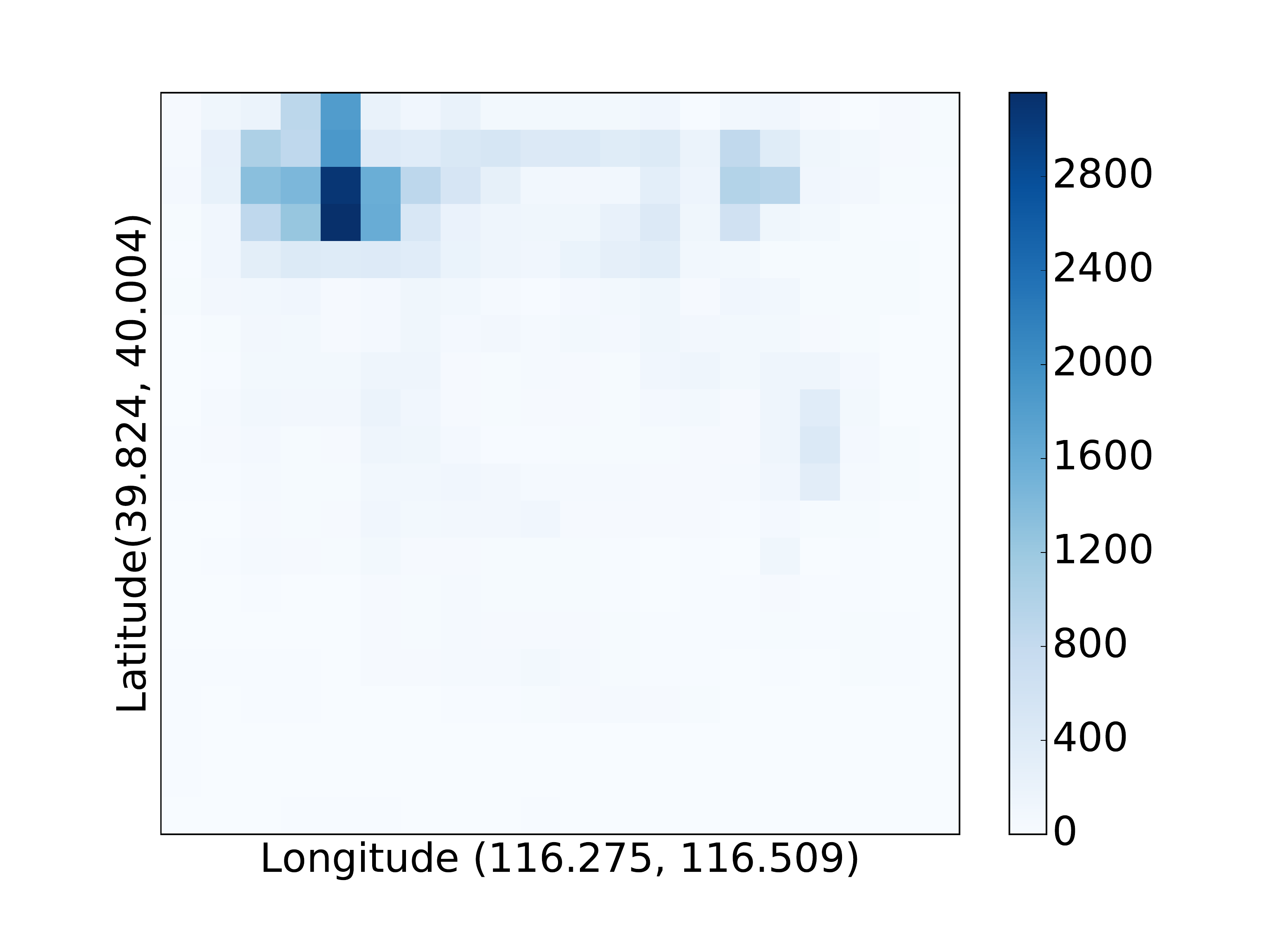}}
	\subfloat[\small Cabspotting (San Francisco).\label{fig:san_zoom_heatmap}]{
		\includegraphics[scale = .13]{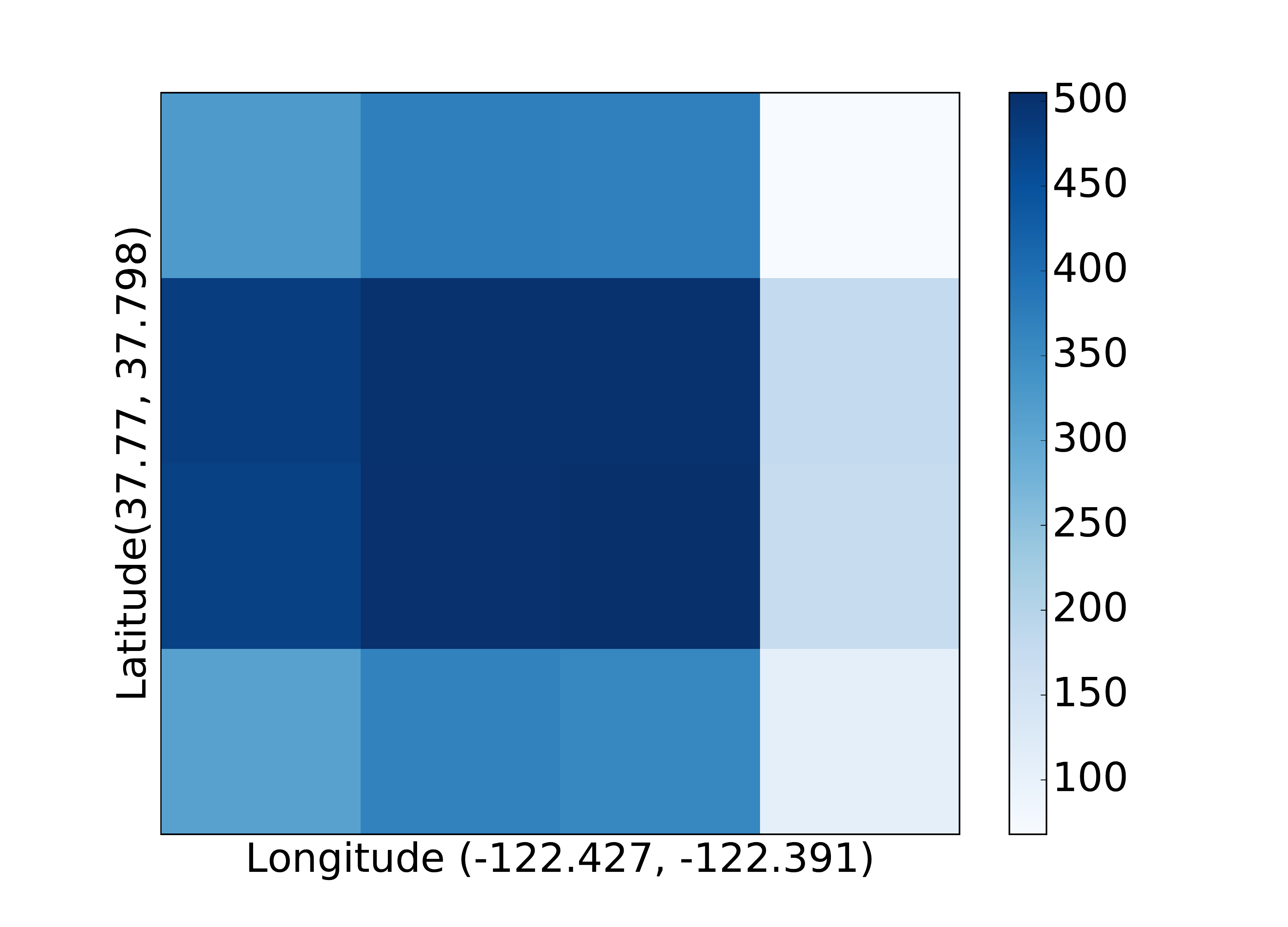}}
	\caption{The density of users' spatial regions for the selected area per dataset, which measure $20km*20km$, $20km*20km$, $3.2km*3.2km$ respectively.}
	\label{fig:datasets_heatmaps}
\end{figure*}

\subsection{Pre-processing}

\begin{figure*}
	\centering
	\subfloat[\label{fig:san_point}]{
		\includegraphics[scale = .15]{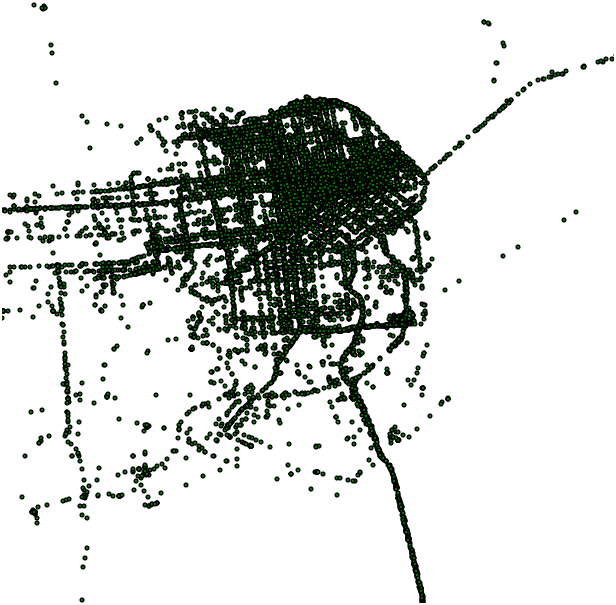}}\hfill
	\subfloat[\label{fig:san_kde}]{
		\includegraphics[scale = .22]{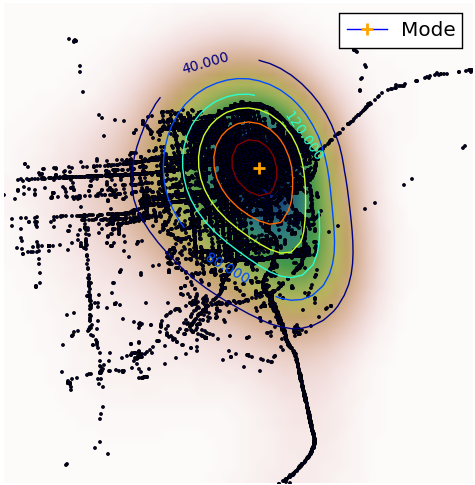}}\hfill
	\subfloat[\label{fig:san_zoom}]{
		\includegraphics[scale = .2]{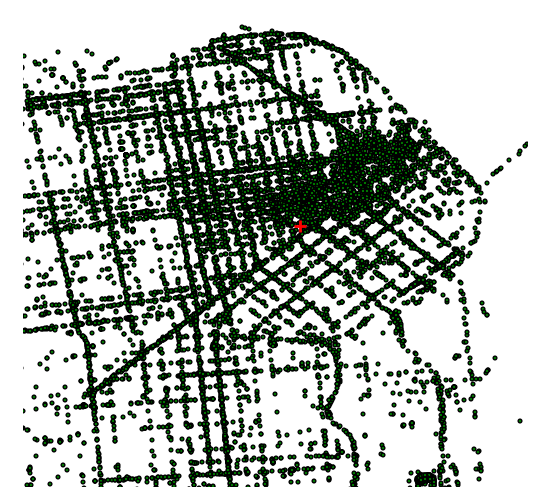}}\hfill
	\subfloat[\label{fig:san_convex}]{
		\includegraphics[scale = .14]{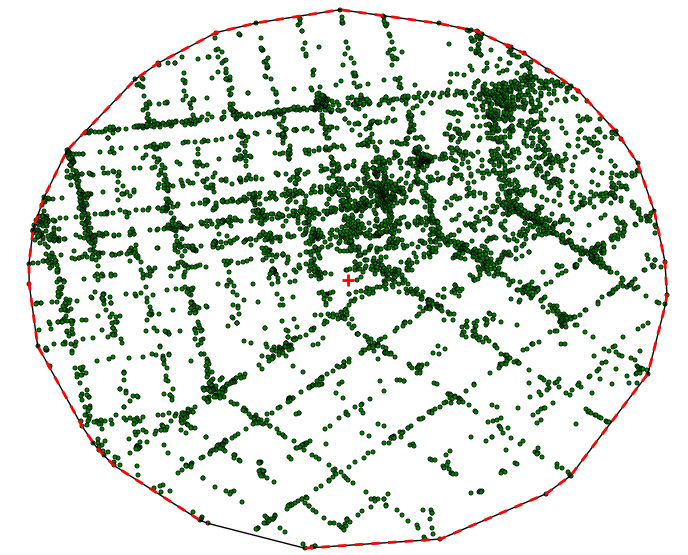}}
	\caption{Pre-processing in experimental setup: Computing the KDE and mode for a set of GPS points, then convex hull. Based on a sample of one cab's GPS points in San Francisco, from Cabspotting.}
	\label{fig:san_sample}
\end{figure*}

We pre-process each dataset to extract convex planar bodies, representing regions where users mostly frequent. This simulates a real application where extraction might be conducted at the end point. For instance, in some applications, \eg in a fitness tracker, users can determine the area in which they usually locate their workouts, in order to obtain a desired service. We conduct the following steps, which represent just one approach to creating convex regions of high visitation. 

\begin{compactitem}
	\item Fit a kernel density estimate (KDE) and consequently take the mode of each user's set of GPS points;
	\item Take k-nearest neighbours (k-NN) points to the mode, \eg for GeoLife, 8 hours corresponds to $k=5760$. If the number of GPS points are less than $k$ we take all points;
	\item Check if all the points are within the defined $B$ diameter, otherwise discard outliers; and
	\item Compute the convex hull of remaining points to create a convex planar body representing an area of frequent visitation.
\end{compactitem}

We use standard libraries from the Scipy package~\cite{Scipy} to compute the kNN and convex hull. To deal with geographic coordinates, Euclidean computations do not directly apply, such as to calculate a distance between two points (here: the opposite corners of a bounding box of $B$ diameter). We refer the interested reader to~\cite{Iliffe2008datums}.

Figure~\ref{fig:san_sample} demonstrates the trajectory of a cab in San Francisco~\ref{fig:san_point}, taken from the Cabspotting project. In this picture (\cf  Figure~\ref{fig:san_kde}), the level sets within the contour lines are convex, and we could have picked these for our convex planar body. But in general, level sets are not convex. Our approach generates a convex approximation. As depicted in  Figure~\ref{fig:san_zoom}, cab GPS points in this dataset are dense and concentrated in a specific area.  Figure~\ref{fig:san_convex} illustrates the extracted convex body. 

After pre-processing, we create histogram counts per convex body, to construct the Euler histograms as our baseline approach and as the basis for our other algorithms. 

For constrained inference, we used the Gurobi optimisation software package~\cite{Gurobi} which implements dual simplex and barrier algorithms to solve \LP, with concurrent optimization. 
We next explain how to choose parameters, then describe our evaluation metrics.

\subsection{Parameter Settings}\label{sec:exp-param}

Initial settings for Beijing with four parameters $A$ (area side length), $d$ (cell size), $B$ (bounded diameter), $\epsilon$ are $ 20km $, $ 1km $, $ 2km $ and $ 1 $ respectively. These settings are applied on T-Drive, and GeoLife1.3 datasets. With regard to San Francisco, Cabspotting dataset, area size is $ 3.2km \times 3.2km $, and cell size is $ 0.8km $ with the remaining parameters the same. The density of users' spatial regions per dataset in a selected spatial partition are illustrated in Figure~\ref{fig:datasets_heatmaps}. As shown in Figures~\ref{fig:tdrive_heatmap_mesh},~\ref{fig:geo13_heatmap_mesh}, T-Drive and GeoLife reflects a different distribution of users' spatial regions, and even for the selected area of San Francisco some partitions are more dense, Figure~\ref{fig:san_zoom_heatmap}.

Table~\ref{table:experiment} demonstrates our experimental parameter values, not including the experiments in Sections~\ref{sec:histograms_diff}--\ref{sec:timing}. As demonstrated, bold parameters are varying.  

Even though the literature on point data~\cite{Cormode12PSD,QardajiYL13} tends to use only specific QR sizes, we vary the QR parameter over the entire range of the area size to more fully evaluate our technique. For experiments where we compare histograms, the $A/d$ ratio, which defines the number of grid cells for each axis, has been kept constant for all datasets (\cf  Sections~\ref{sec:histograms_diff}--\ref{sec:timing}).

\subsection{Evaluation Metrics}\label{sec:evaluation}
Apart from the varying parameter, we keep all other parameters fixed to compute the \term{median relative error} as an empirical measure of utility, as is standard~\cite{Cormode12PSD,QardajiYL13}. We repeat each of the experiments 100 times and compute median relative error. The baseline approach is \Eu as it provides exact answers. Algorithms \DP, \LP, \R that are privacy-preserving, are compared to \Eu. 
Another evaluation metric is the percentage of the \emph{differences between Euler histograms} and the \DP, \LP, and \R approaches over the \emph{difference between Euler histograms} and the \DP (\emph{relative error to} \DP). Here, the $L_1$-norm is adopted.

We also compute the number of times each constraint has been violated in each technique compared to \LP and \R which are the consistent techniques, as well as \Eu as our baseline approach, \emph{consistency constraints violation}. 
Furthermore, we compute the \emph{running time} for each algorithm (\cf Section~\ref{sec:timing}).

\begin{table}
	\centering
	\caption{Experimental settings. This table shows the range of parameters, bolded are those that are varying.}\label{table:experiment}
	\bigskip
	\scalebox{.95}{
		\begin{tabular}{| c | c | c | c | c | c | c |}
			\hline
			\textbf{Dataset} & \textbf{Cell Size (d)} & \textbf{B} &  \textbf{Area Size (A)} & \textbf{A/d} & \textbf{QR Size/Shape} & $\epsilon$\\
			\hline
			\hline
			T-Drive & 1km & 2km &  20km*20km & 20 & \textbf{1-10\%} & 1\\
			\hline
			T-Drive & 1km & 2km &  20km*20km & 20 & \textbf{10-100\%} & 1\\
			\hline
			T-Drive & \textbf{0.66,1,2km} & 2km &  20km*20km & \textbf{30,20,10} & 1\% & 1\\
			\hline
			T-Drive & 2km & 2km &  20km*20km & 10 & 1\% & \textbf{0.1,0.4,0.7,1}\\
			\hline
			GeoLife1.3 & 1km & 2km &  20km*20km & 20 & \textbf{1-10\%} & 1\\
			\hline
			GeoLife1.3 & 1km & 2km &  20km*20km & 20 & \textbf{10-100\%} & 1\\
			\hline
			Cabspotting & 0.8km & 2km &  3.2km*3.2km & 4 & \textbf{10-100\%} & 1\\
			\hline
	\end{tabular}}
\end{table}

\begin{figure}
	\begin{minipage}[t]{1\textwidth}
		\centering
		\subfloat[QR Size (1-10\% of Total Area).\label{fig:tdrive_qr_1-10}]{
			\includegraphics[scale = .2]{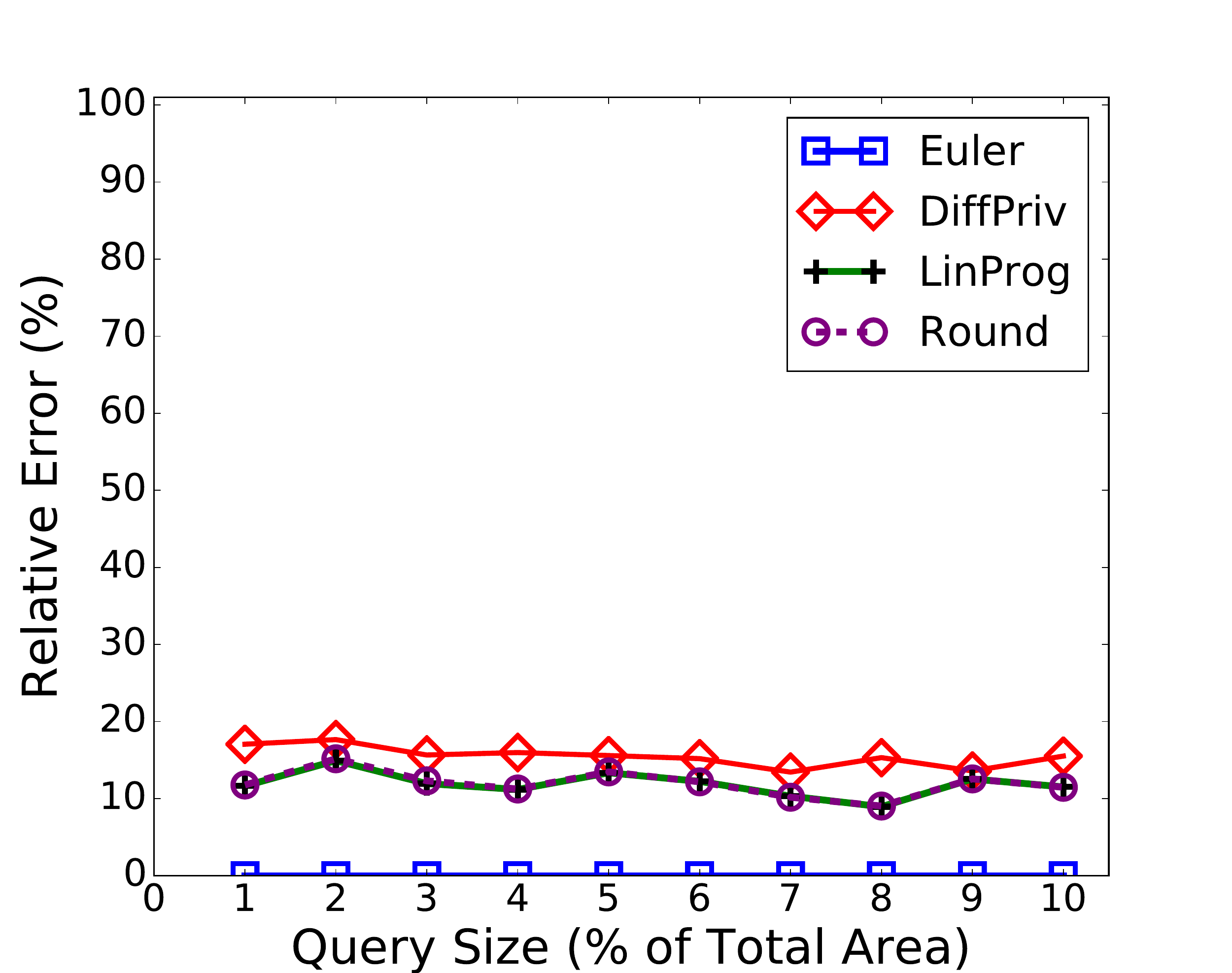}}
		\subfloat[QR Size (10-100\% of Total Area).\label{fig:tdrive_qr_10-100}]{
			\includegraphics[scale = .2]{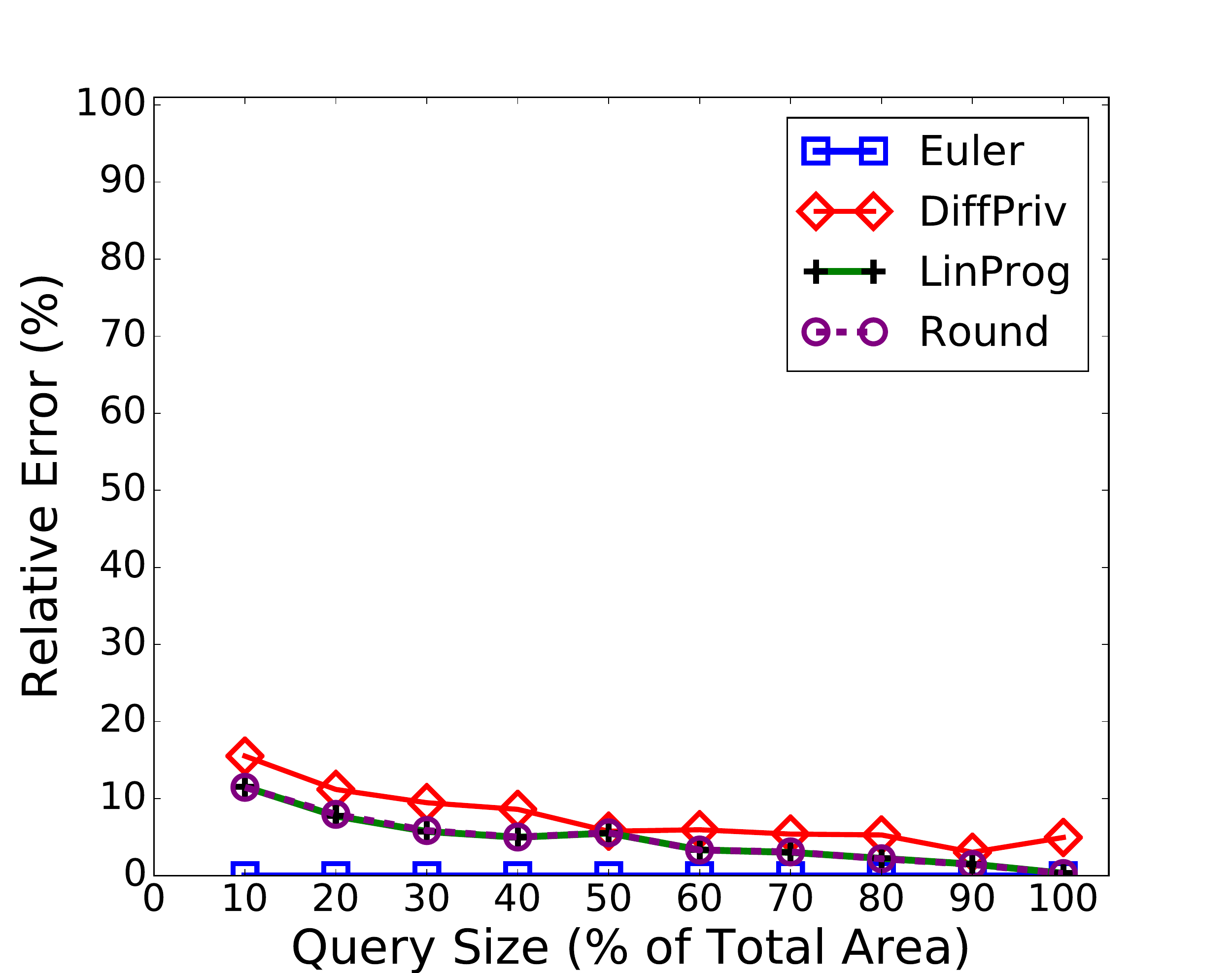}}
		
		\subfloat[Various QR Shapes (smaller).\label{fig:tdrive_QR_shapes_error_bar_small}]{
			\includegraphics[scale = .2]{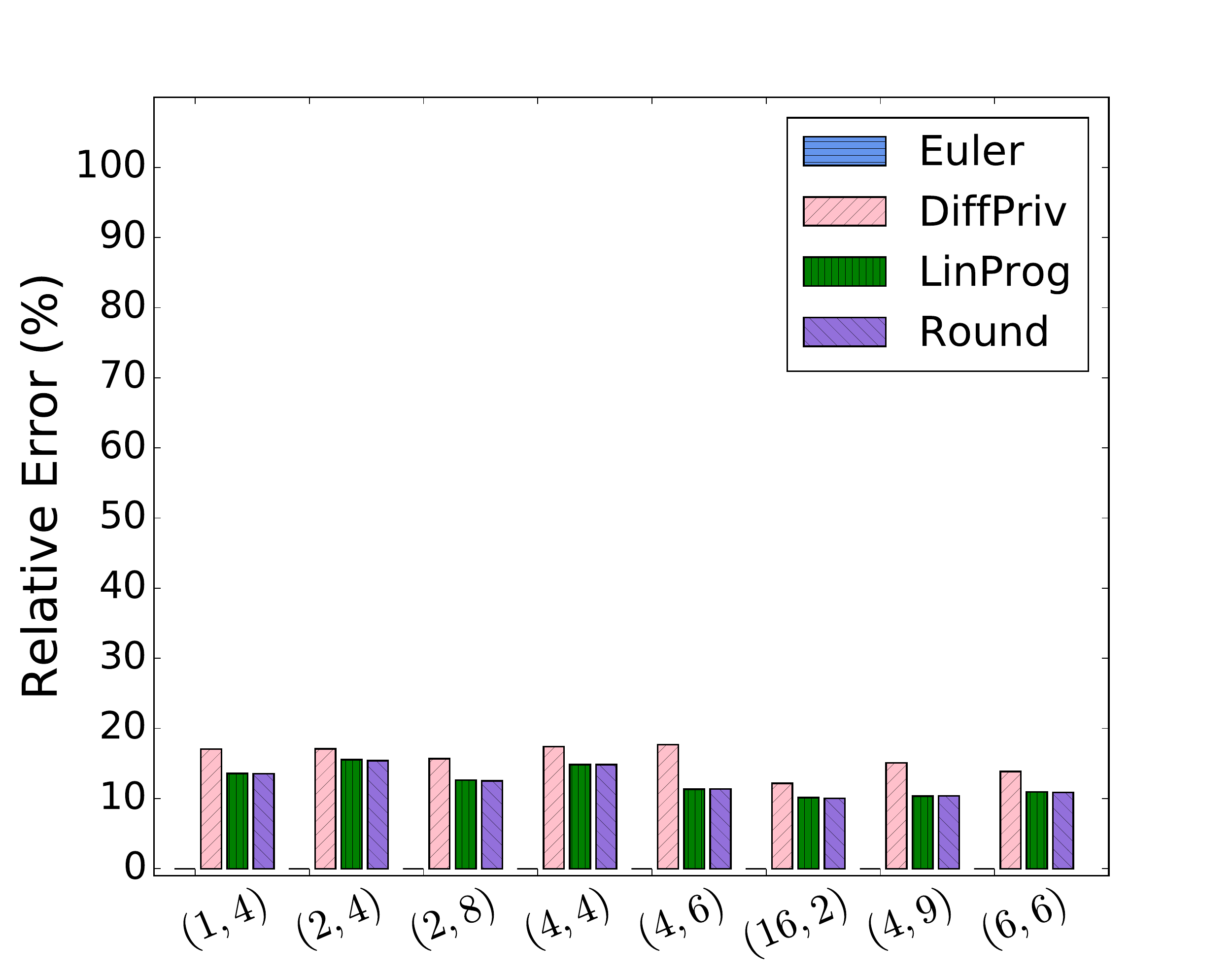}}
		\subfloat[Various QR Shapes (larger).\label{fig:tdrive_QR_shapes_error_bar_large}]{
			\includegraphics[scale = .2]{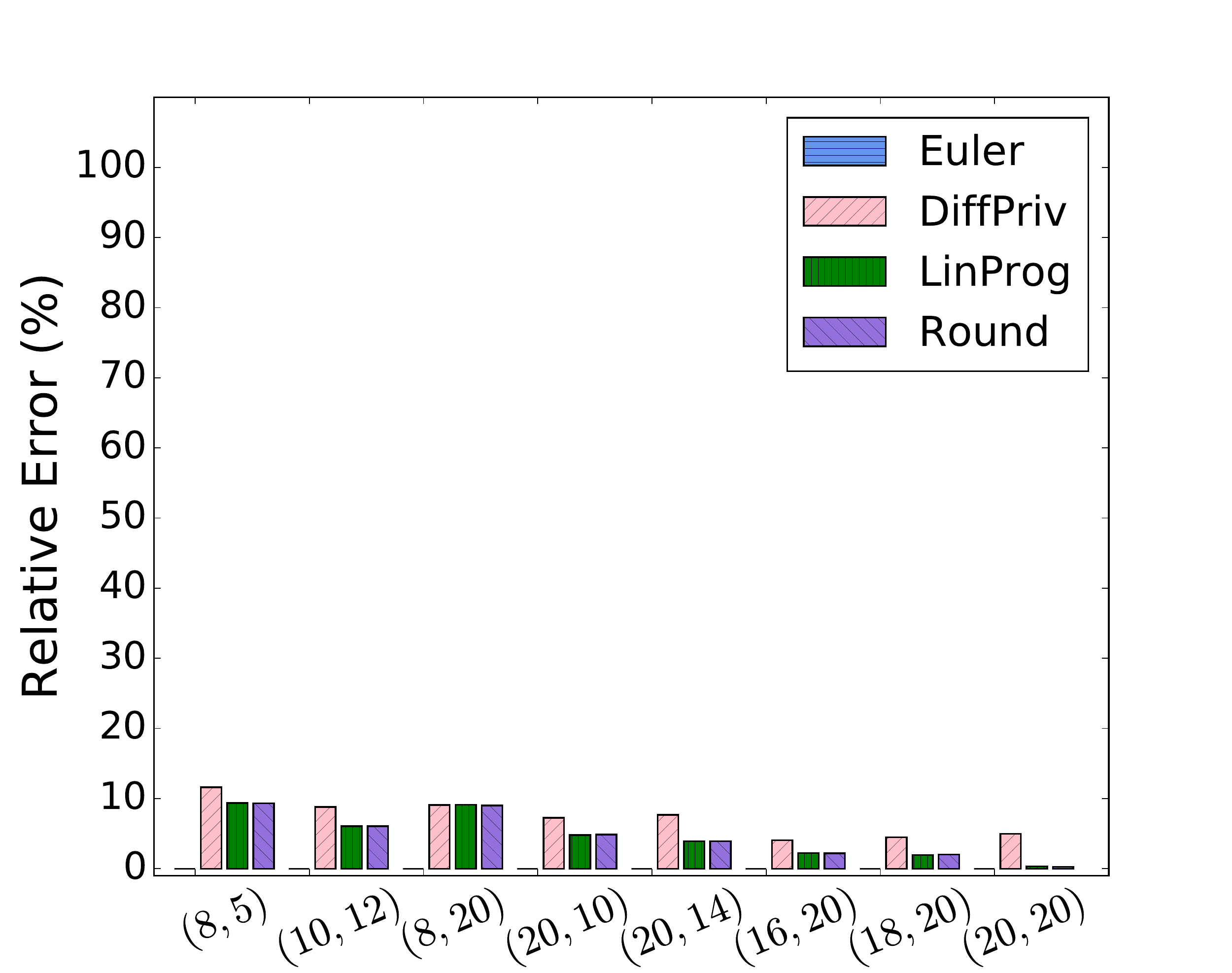}}
		\caption{Median relative error per query size and shape for T-Drive dataset.}
		\label{fig:tdrive_qr_error}
	\end{minipage}
\end{figure}

\begin{figure}
	\begin{minipage}[t]{1\textwidth}
		\centering
		\subfloat[QR Size (1-10\% of Total Area).\label{fig:geo13_qr_1-10}]{
			\includegraphics[scale = .2]{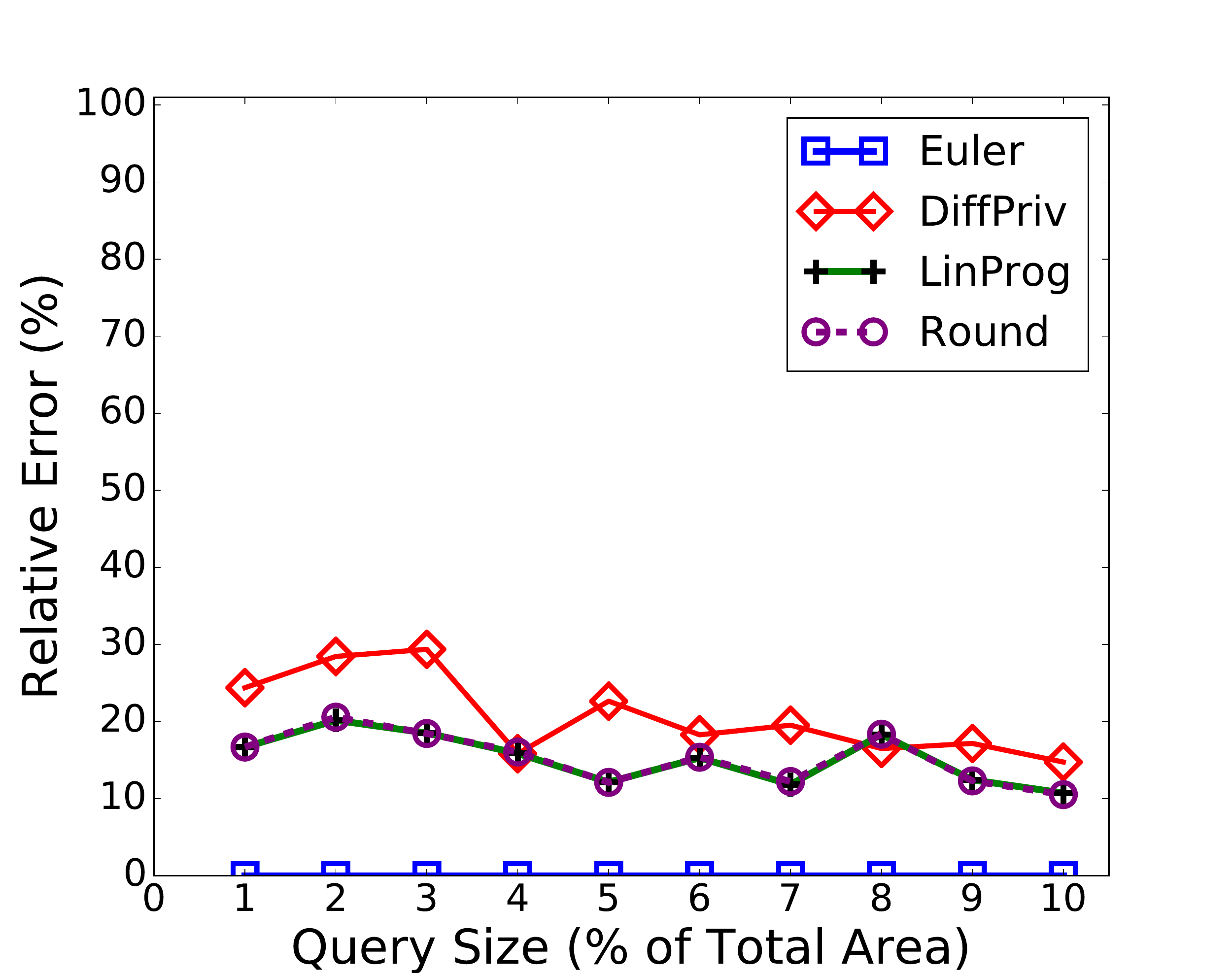}}
		\subfloat[QR Size (10-100\% of Total Area).\label{fig:geo13_qr_10-100}]{
			\includegraphics[scale = .2]{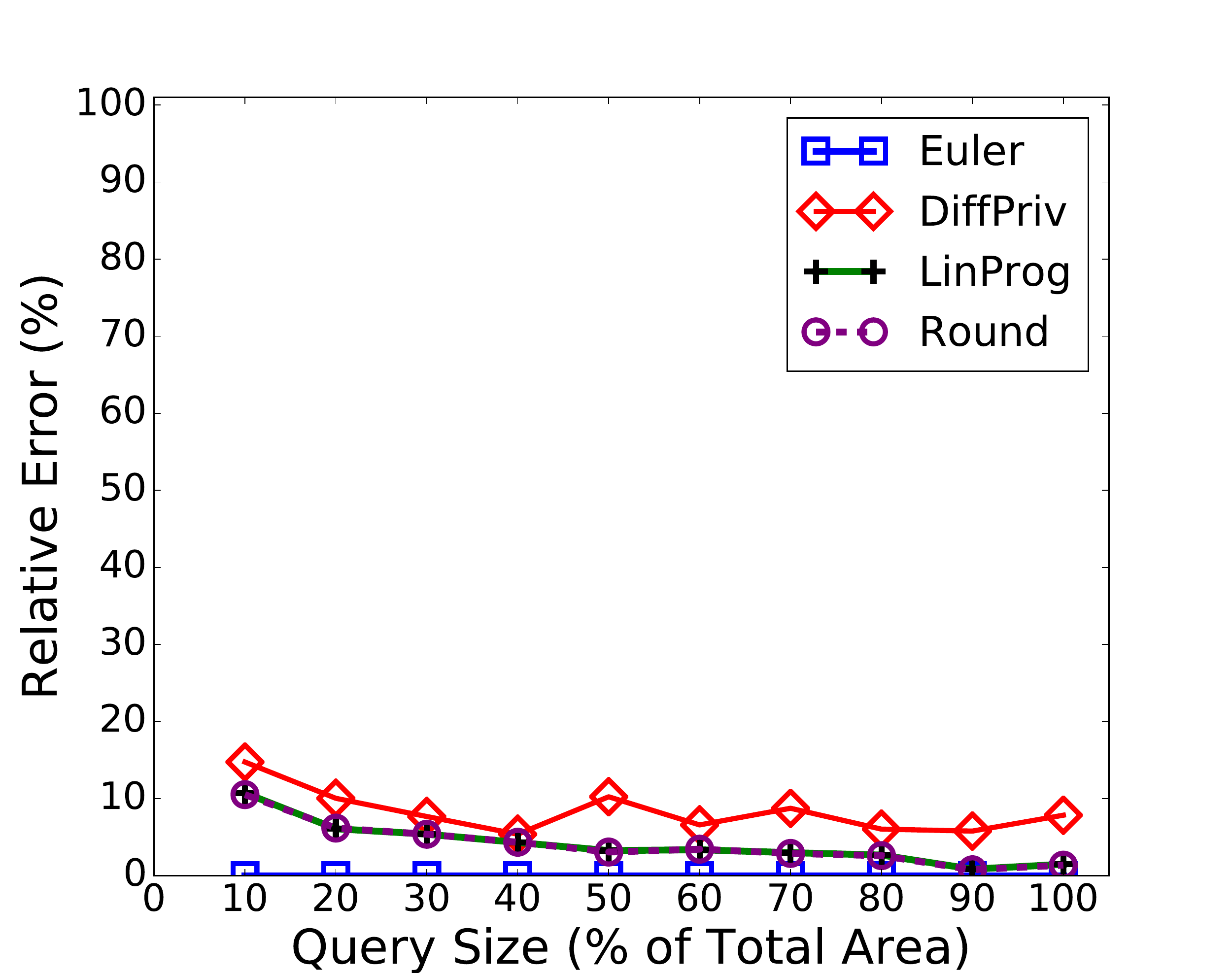}}
		
		\subfloat[Various QR Shapes (smaller).\label{fig:geo13_QR_shapes_error_bar_small}]{
			\includegraphics[scale = .2]{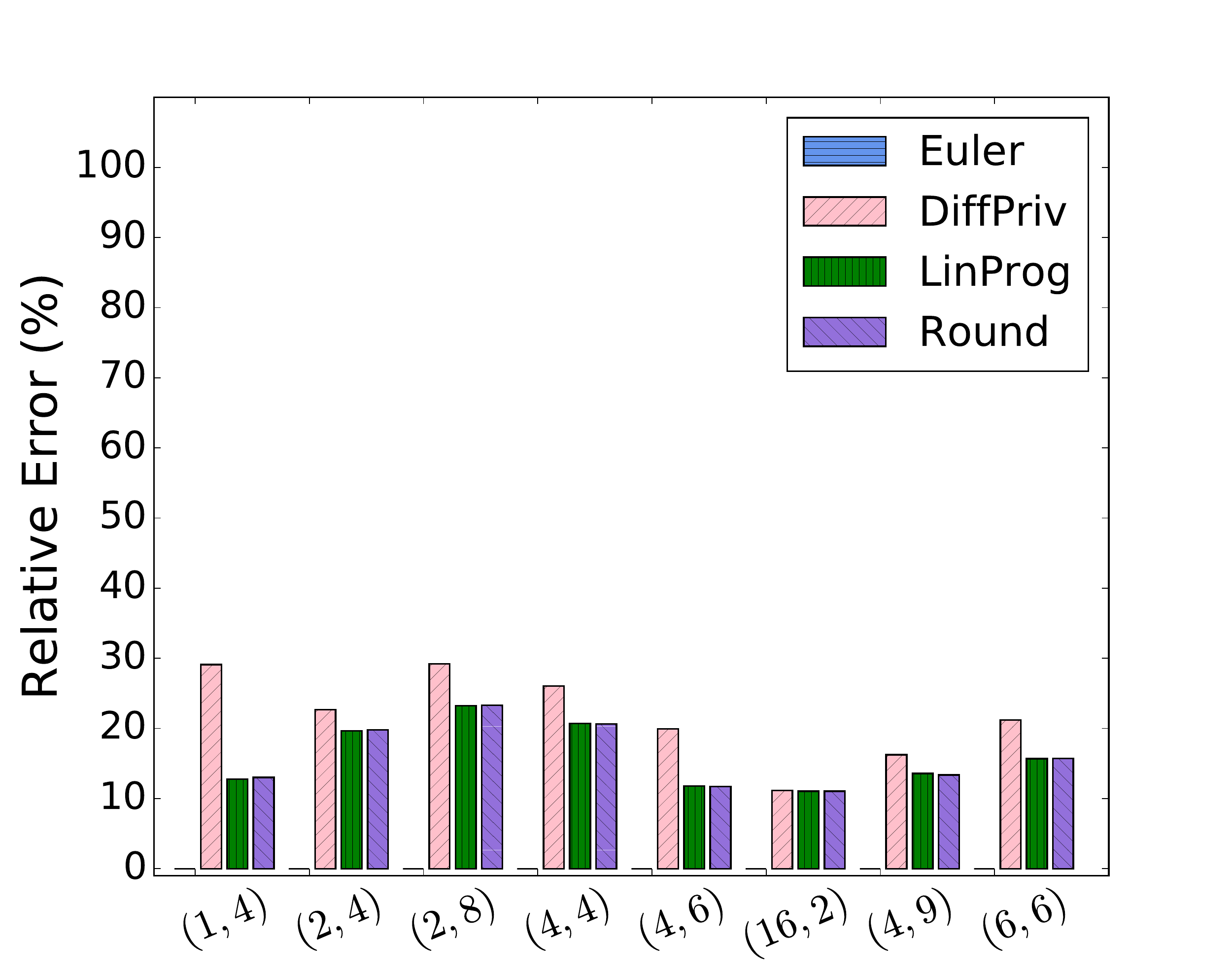}}
		\subfloat[Various QR Shapes (larger).\label{fig:geo13_QR_shapes_error_bar_large}]{
			\includegraphics[scale = .2]{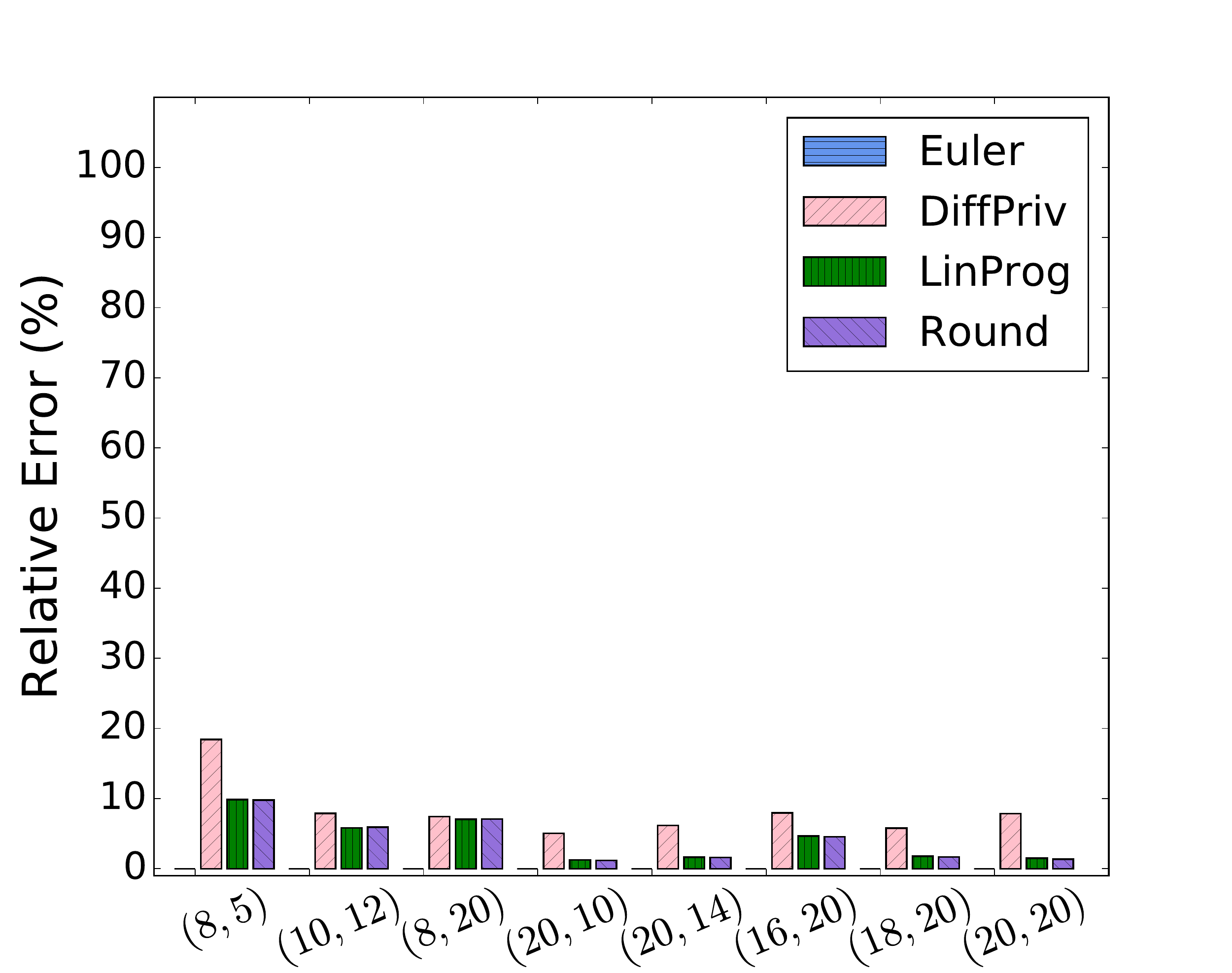}}	
		\caption{Median relative error per query size and shape for GeoLife1.3 dataset.}
		\label{fig:geo13_qr_error}
	\end{minipage}
	\begin{minipage}[t]{1\textwidth}
		\centering
		\subfloat[Various QR Shapes.\label{fig:san_QR_shapes_error_bar}]{
			\includegraphics[scale = .2]{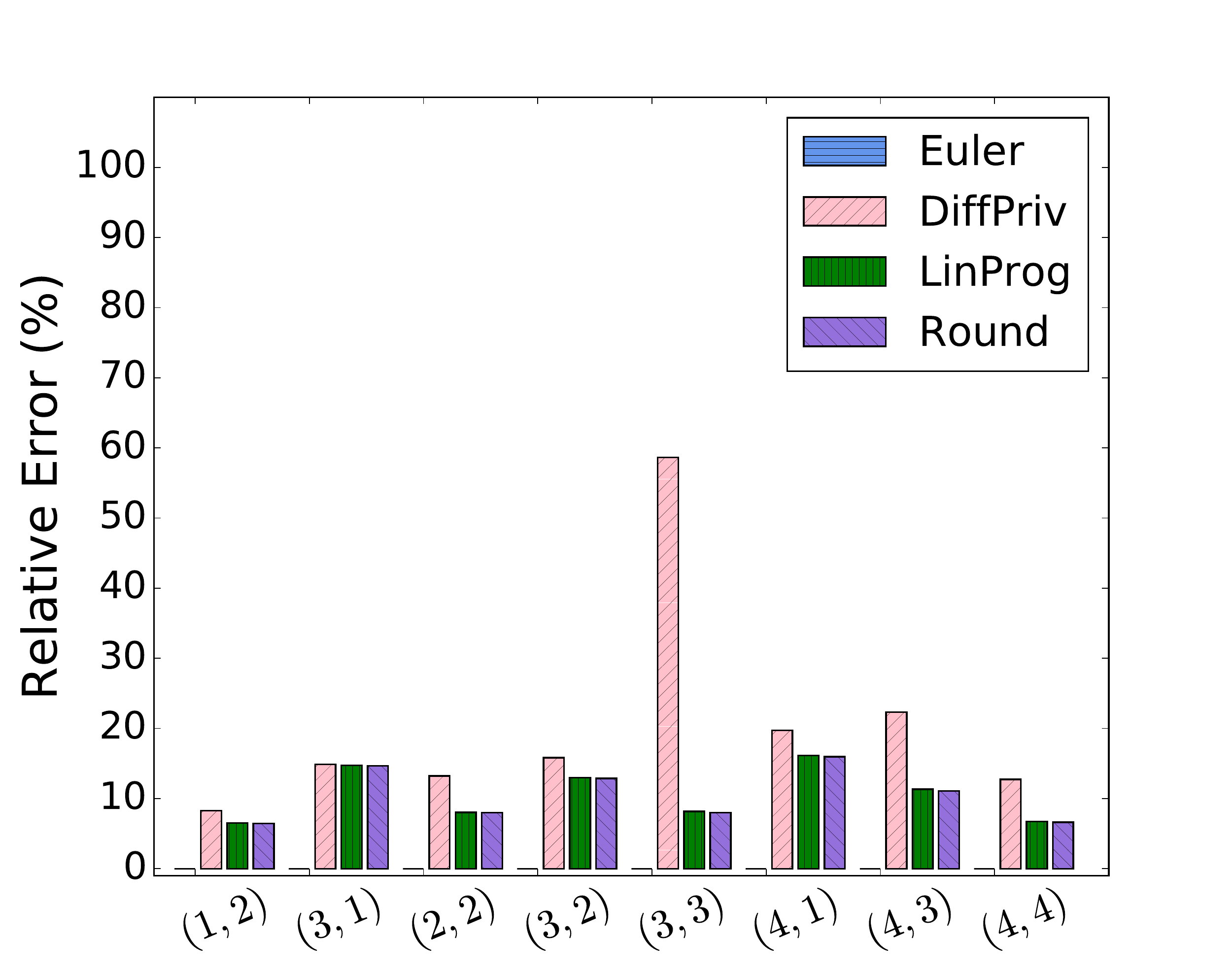}}
		\subfloat[QR Size (10-100\% of Total Area).\label{fig:san_qr_10-100}]{
			\includegraphics[scale = .2]{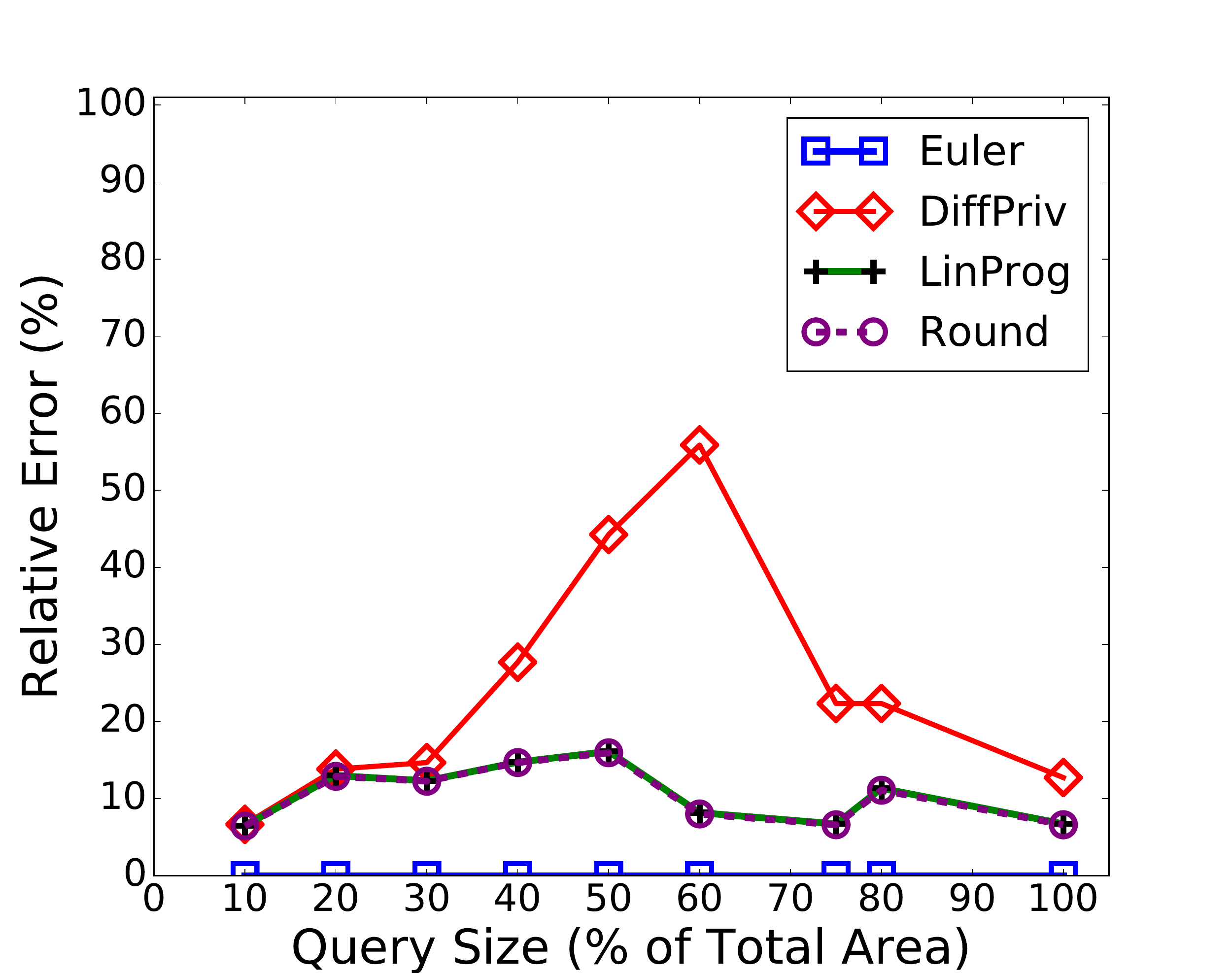}}
		\caption{Median relative error per query size and shape for Cabspotting dataset.}
		\label{fig:san_qr_error}
	\end{minipage}
\end{figure}

\subsection{Varying Query Rectangle Size}\label{sec:qr_size}

In this section we compute the median relative error on all datasets, representing diversity in terms of sparsity, density and concentration, to demonstrate effect on accuracy. 
We fix every parameter, except QR size to run a range query on various sizes, with varying position on the partitioned map, based on definition of a QR as a union of grid cells. Range queries are varied from 1 to 10 and 10 to 100 percent of the total area size of the respective city. The results for various sizes as well as shapes of a range query are shown in  Figures~\ref{fig:tdrive_qr_error}--\ref{fig:san_qr_error}. Various parameters can affect the response to a QR, including shape of a QR, size of a QR, whether convex bodies are sparse in the space or dense, or if they are concentrated or not.  Furthermore, the computed global density (\cf Lemma~\ref{lem:gs}) differs across dataset settings, \eg $ 25 $ for both T-Drive and GeoLife datasets, and $49$ for Cabspotting, and this value also affects the results.
The similarity between T-Drive and Cabspotting is that both record taxi driver movements; but a difference is that the former is not concentrated on a specific area while the latter is. In GeoLife1.3 the convex bodies are more dense, having a large number of trajectories. 

As depicted in  Figure~\ref{fig:tdrive_qr_error} for the T-Drive dataset, since the data is more evenly distributed the error is very low for larger QR sizes (Figure~\ref{fig:tdrive_qr_10-100}), and is less than $20\%$ for smaller QRs (Figure~\ref{fig:tdrive_qr_1-10}). A variety of QR shapes for the smaller sizes (Figure~\ref{fig:tdrive_QR_shapes_error_bar_small}), and larger ones (Figure~\ref{fig:tdrive_QR_shapes_error_bar_large}) are depicted accordingly. For instance, 1\% QR in a $ 20 \times 20 $ partitioned-map of Beijing city could be $ (1,4) $, $ (2,2) $, $ (4,1) $ geometries, where the first coordinate represents the number of rows and the second represents number of columns. Compared to GeoLife1.3 (Figure~\ref{fig:geo13_qr_error}), since trajectories are more focused on some area (\cf Figure~\ref{fig:geo13_heatmap_mesh}), the error increases by decreasing QR size (Figure~\ref{fig:geo13_qr_1-10}). 

With regard to the Cabspotting dataset (Figure~\ref{fig:san_qr_error}), some parts of the selected area are sparser which consequently affects the result of \DP. Specifically for the QR shape of $(3,3)$ (Figure~\ref{fig:san_QR_shapes_error_bar}) and the QR sizes of 50\% and 60\% (Figure~\ref{fig:san_qr_10-100}), such QRs contain dense and sparse cells. This results in larger errors. However for larger QRs, errors cancel each other out due to the Euler formula, Equation~\eqref{eq:euler}. In all cases, \LP and \R reduce the errors, and provide a high level of accuracy. 
Since the number of spatial partitions for the chosen area is smaller than the other datasets, only QR sizes and shapes between 10\%--100\% are shown in Figures~\ref{fig:san_QR_shapes_error_bar} and \ref{fig:san_qr_10-100}. The QR errors for the smaller sizes 1\%--9\% are less than 10\%. 

\LP and \R provide similar results, and as discussed in Section~\ref{sec:approach}, the difference is the covertness property of \R. There is inconsistency in the \DP histogram results (see Section~\ref{sec:consistency_violation}). Providing consistency, through the \LP and \R techniques, can improve accuracy (\cf Sections~\ref{sec:cell_size},~\ref{sec:epsilon}). 

For the remainder of the experiments for varying other parameters, we focus results on T-Drive dataset, and the 1\% QR size as a conservative representative, since it incurs higher error. 

\begin{figure}
	\begin{minipage}[t]{0.49\textwidth}
	\centering
	\includegraphics[scale=0.25]{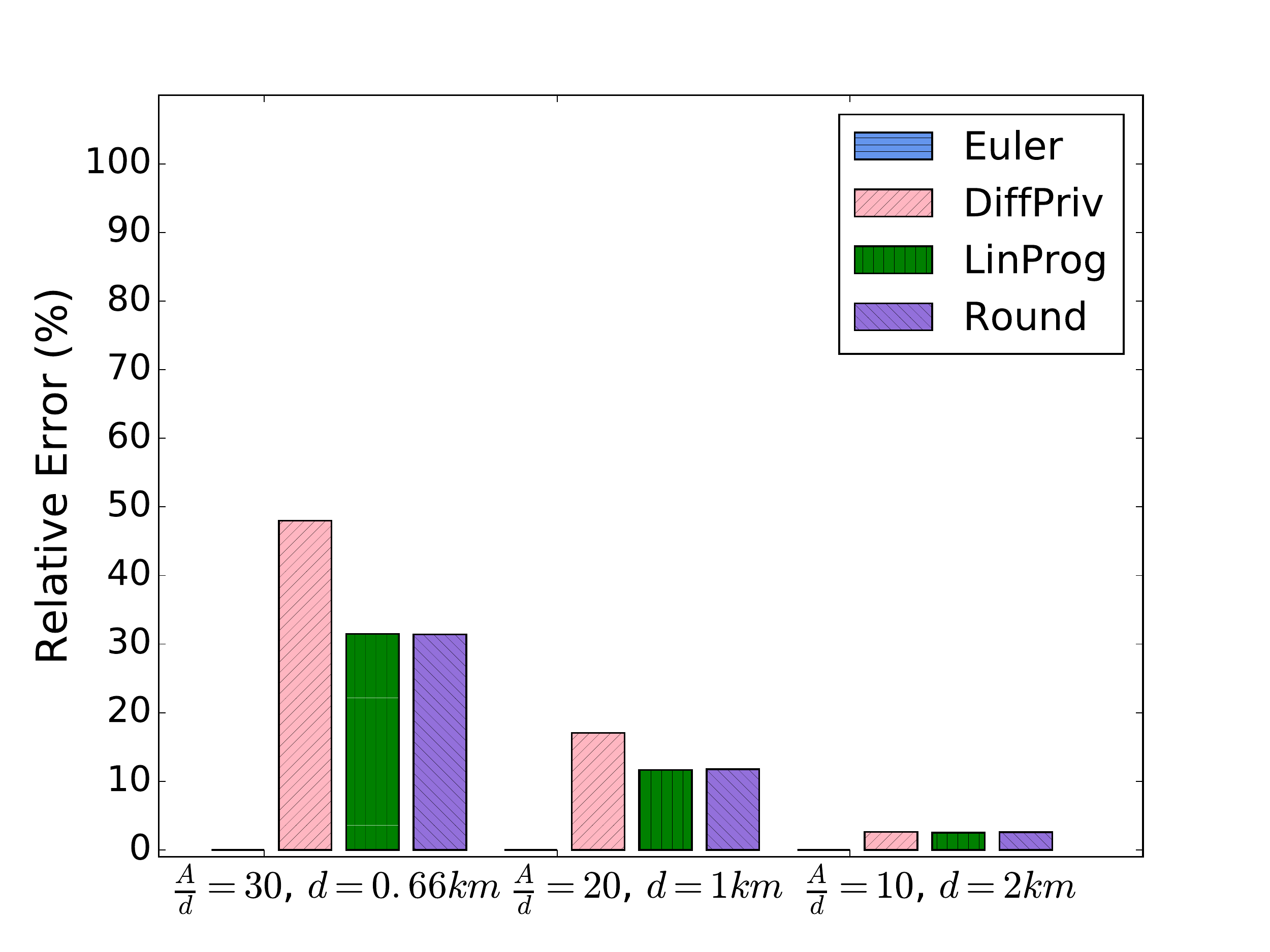}
	\caption{Varying area size/cell size ratio for T-Drive dataset.}
	\label{fig:cell_size}
	\end{minipage}\hfill
	\begin{minipage}[t]{0.49\textwidth}
	\centering
	\includegraphics[scale = 0.25]{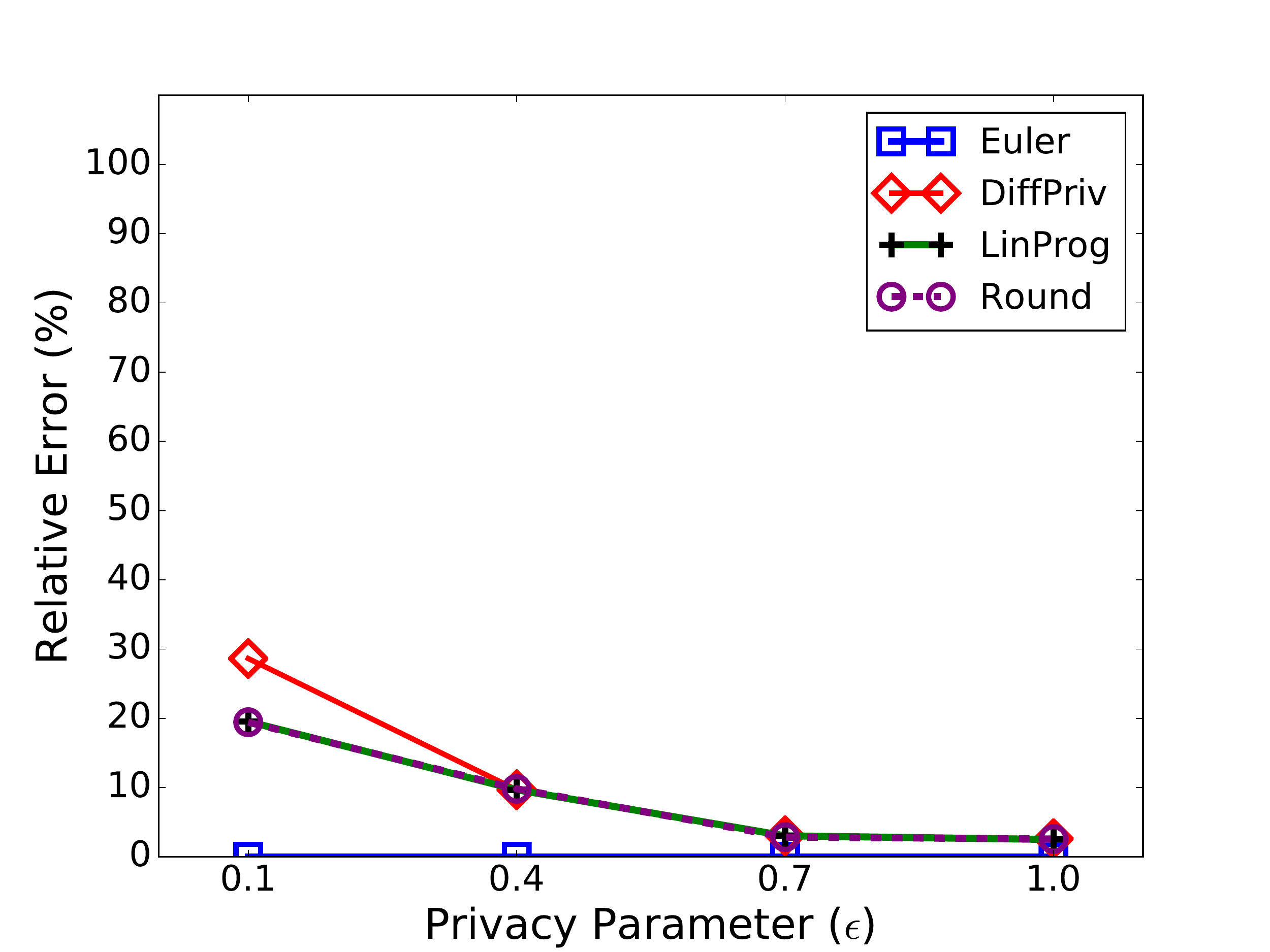}
	\caption{Varying privacy parameter for T-Drive dataset.}
	\label{fig:epsilon}
	\end{minipage}
\end{figure}

\subsection{Varying Area Size/Grid Cell Size Ratio}\label{sec:cell_size}

We vary the area size (A) over grid cell size (d) ratio and compute the median relative error for QR taken as 1\% of total area of T-Drive dataset. The area size for this dataset is $ 20km \times 20km $. By increasing the cell size, we expect that the accuracy improves, as demonstrated in Figure~\ref{fig:cell_size}. We have fixed the QR as 1\%, and varied the size of the grid cell in a range 0.66km, 1km, and 2km to yield the ratios of 30, 20, and 10 respectively. As shown, by increasing the grid cell size the accuracy increases. As illustrated in Figure~\ref{fig:cell_size}, as we decrease the grid cell size, the error increases due to higher values of global sensitivity for smaller cell sizes: $ 49 $, $ 25 $, $ 9 $ are the global sensitivity (GS) values for $ 0.66km $, $ 1km $, and $ 2km $ cell sizes respectively.
If we wish to decrease $d$ without incurring reduced accuracy, our theoretical results suggest that we should also decrease $B$ and $A$.

\subsection{Varying Privacy Parameter $ \epsilon $}\label{sec:epsilon}
We apply a similar procedure to vary the privacy parameter across values 0.1, 0.4, 0.7, and 1 with fixed QR of 1\% of the total area $ 20km \times 20km $, and cell size $ 2km $.  The effect of increasing $\epsilon$ on accuracy is depicted in Figure~\ref{fig:epsilon}. Decreasing the $\epsilon$ value from $1$, will increase the scale parameter of Laplace distribution (added noise to the counts) from $ 9 $ to $ 90 $ for $ \epsilon = 0.1 $, and this affects the accuracy of the result.
To keep accuracy relatively constant when reducing $\epsilon$, the third party can vary other parameters.

\subsection{Difference on Histograms}\label{sec:histograms_diff}
Computing the differences on histograms and their relative error to \DP show that \LP and \R are superior in terms of having less difference, while still being private. For this part of the experiment, in order to make the dataset histogram differences comparable, we kept the $A/d$ ratio fixed as discussed in Section~\ref{sec:exp-param}, for San Francisco 3.2km/0.16km and for Beijing 20km/1km.

Figure~\ref{fig:histograms_diff} depicts this comparison for all datasets, showing that on the first  concentrated dataset, \DP has a considerable difference with \LP and \R (\cf Cabspotting). This difference decreases for the relatively evenly distributed datasets (\cf GeoLife and T-Drive). \LP and \R have similar differences with \Eu.
\begin{figure}
	\begin{minipage}[t]{0.49\textwidth}
	\centering
	\includegraphics[scale = .25]{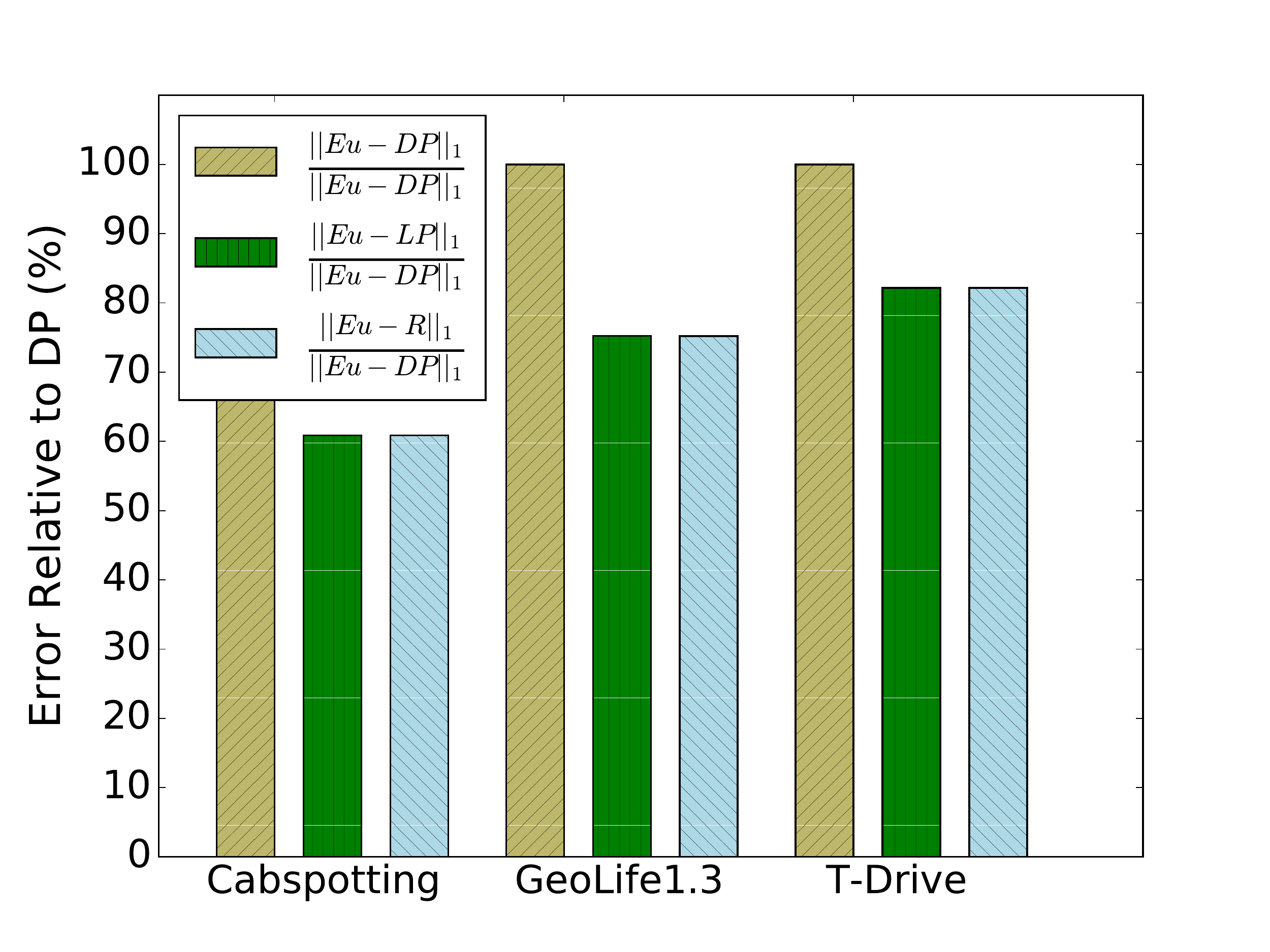}
	\caption{Differences on histograms for all datasets.}
	\label{fig:histograms_diff}
\end{minipage}\hfill
\begin{minipage}[t]{0.49\textwidth}
	\centering
	\includegraphics[scale = .3140]{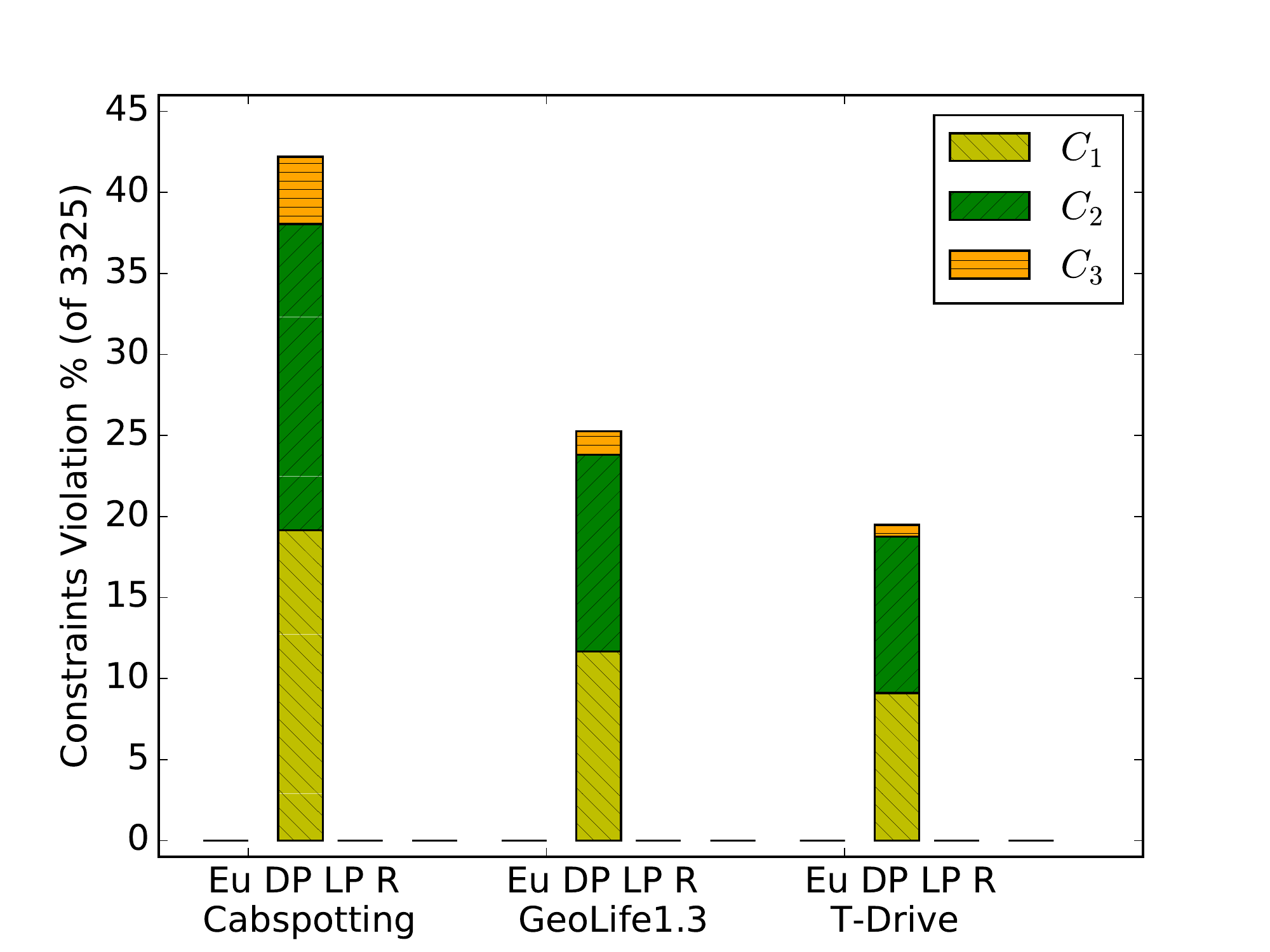}\\[0.2em]
	\caption{Consistency constrains violations for all datasets.}
	\label{fig:consistency_bar}
\end{minipage}
\end{figure}

\subsection{Consistency Constraint Violations}\label{sec:consistency_violation}
In Figure~\ref{fig:consistency_bar} the percentage of violations of constraints $C_1,C_2,C_3$ is depicted for our datasets. The total number of constraints for each of the datasets is 3325 (as we held $n$ fixed), in which 1520 are for $C_1$, 1444 for $C_2$ and 361 for $C_3$. Approximately the same proportion of the constraints are violated in each dataset, and $C_3$ is less than the other constraints, therefore it is not considerably violated. As we decrease the size of the grid cell to 0.16km in Cabspotting, the global sensitivity (\cf Definition~\ref{def:gs}) increases to 729, therefore it has a greater percentage of violation compared to the other datasets. 

\subsection{Running Time}\label{sec:timing}
Figure~\ref{fig:timing} shows running times for all datasets of various sizes. As discussed in Section~\ref{sec:exp-param}, we kept the ratio $A/d$ fixed. The running time for all the datasets are approximately similar per technique. The y-axis is in seconds (log-scale) and for the largest dataset GeoLife1.3, the total running time is $ \approx $ 196 seconds. \DP, \LP and \R take less than 1 second for all the datasets. 
\emph{Each of our algorithms are eminently practical to implement and to run.}

\begin{figure}
	\centering
	\includegraphics[scale = .34]{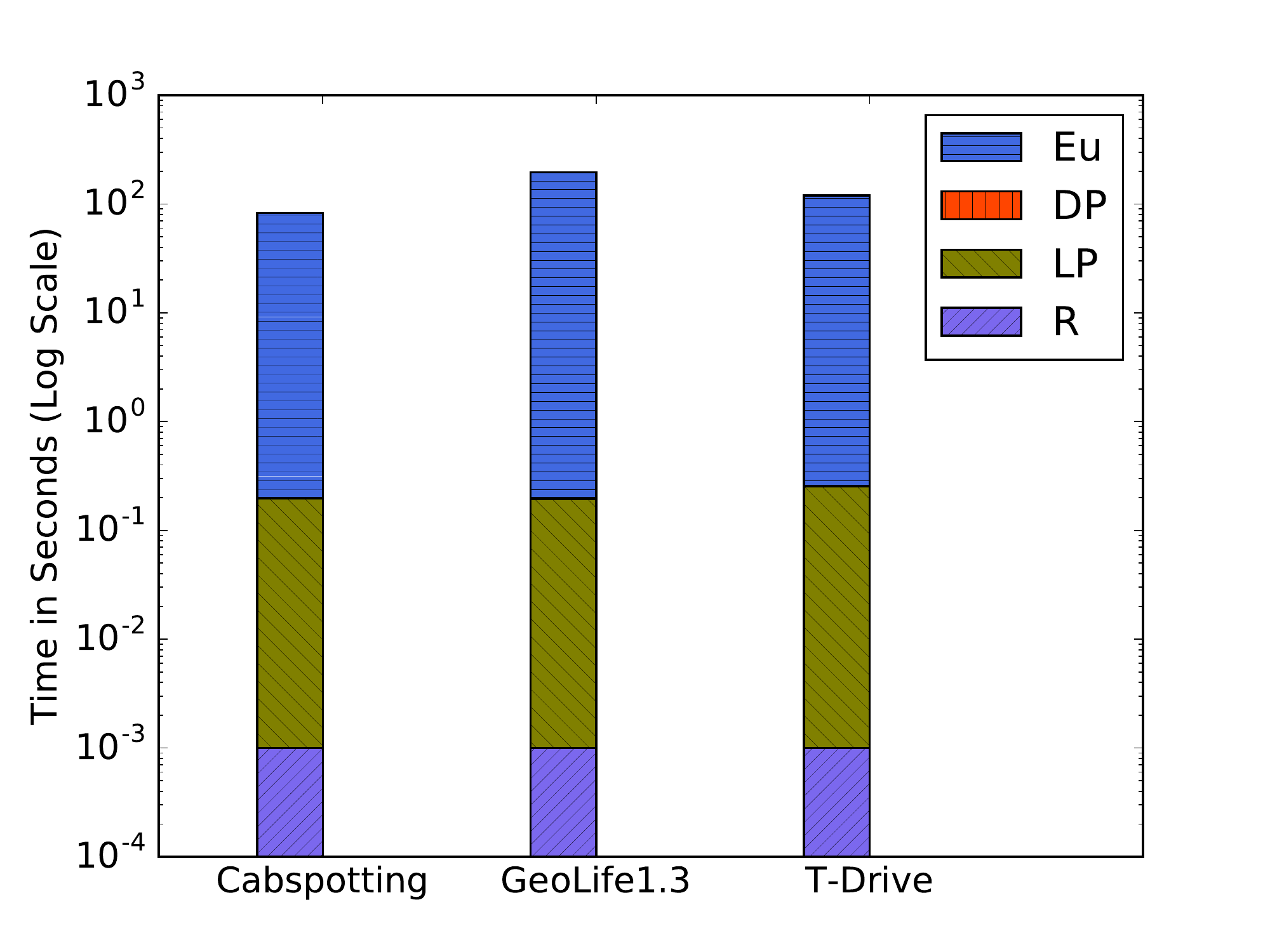}
	\caption{Running time per algorithms for all datasets.}
	\label{fig:timing}
\end{figure}

\section{Concluding Remarks}\label{sec:conclusion}
For the first time we propose a non-interactive differentially-private approach to counting planar bodies representative of users' spatial regions \eg a workout area, areas of customer preference for hotel bookings, or locations of frequent visitation for facility planning.

The key insight of our approach is to leverage Euler histograms for accurate counting, cell perturbations for differential privacy, and constrained inference smoothing to reinstate consistency. Constrained inference often improves utility by cancelling noisy perturbations. Our formulation of constrained inference is a novel constrained application of the robust method of least absolute deviations. Unlike existing constrained inference based on ordinal regression, our formulation precisely matches our privacy-preserving cell perturbation distribution according to maximum-likelihood estimates. 
By optimising for consistency while rounding cell counts, we achieve a covertness property for our counting mechanism: third parties cannot determine that we have perturbed data in the first place.

A full theoretical analysis of utility and differential privacy is complemented by experimental results on three datasets. As demonstrated in the experimental study, uniformly distributed datasets and larger grid partitions result in a better performance. The best practice to select the cell size is the smallest QR that a third party might run on an area to achieve appropriate utility.

Potential directions for future research include utilising adaptive partitioning to have varying partitions sizes according to the dataset distributions to improve the accuracy. The constraints that we have defined for the Euler histogram counts could be potentially more tight to improve utility. Finally, prior public knowledge about true counts could be incorporated into our constrained inference via regularisation that corresponds to Bayesian priors.

\begin{acknowledgements}
This work was supported in part by Australian Research Council DECRA grant DE160100584.
\end{acknowledgements}

\bibliographystyle{spmpsci}
\bibliography{refs}

\end{document}